\documentclass[format=acmsmall, review=false]{acmart}

\usepackage{acme20} %

\usepackage[utf8]{inputenc}
\usepackage{amsmath}
\usepackage{array}
\usepackage{graphicx}
\usepackage[english]{babel}
\usepackage[utf8]{inputenc}
\usepackage{algorithm}
\usepackage[noend]{algpseudocode}

\usepackage{enumitem}
\usepackage{wasysym}
\usepackage{framed}
\usepackage{pgfgantt,rotating}
\usepackage{subfloat}
\usepackage{multirow}
\usepackage{blindtext}

\title{Beyond Optimal Fault-Tolerance}

\author{Andrew Lewis-Pye}
\affiliation{%
  \institution{\department{Department of Mathematics} \institution{London School of Economics} \city{London} \country{UK, a.lewis7@lse.ac.uk}}
}
\email{a.lewis7@lse.ac.uk}

\author{Tim Roughgarden}
\affiliation{%
  \institution{\department{Computer Science Department}
    \institution{Columbia University \& a16z Crypto} \city{New York} \country{USA, tim.roughgarden@gmail.com}}
}
\email{tim.roughgarden@gmail.com}

\date{November 2024}

\begin{abstract} 
One of the most basic properties of a consensus protocol is its fault-tolerance---the maximum fraction of faulty participants that the protocol can tolerate without losing fundamental guarantees such as safety and liveness. Because of its importance, the optimal fault-tolerance achievable by any protocol has been characterized in a wide range of settings. For example, for state machine replication (SMR) protocols operating in the partially synchronous setting, it is possible to simultaneously guarantee consistency against $\alpha$-bounded adversaries (i.e., adversaries that control less than an $\alpha$ fraction of the participants) and liveness against $\beta$-bounded adversaries if and only if $\alpha + 2\beta \leq 1$.

This paper characterizes to what extent ``better-than-optimal'' fault-tolerance guarantees are possible for SMR protocols when the standard consistency requirement is relaxed to allow a bounded number $r$ of consistency violations, each potentially leading to the rollback of recently finalized transactions. We prove that bounding rollback  is impossible without additional timing assumptions and investigate protocols that tolerate and recover from consistency violations whenever message delays around the time of an attack are bounded by a parameter $\Delta^*$ (which may be arbitrarily larger than the parameter $\Delta$ that bounds post-GST message delays in the partially synchronous model). Here, a protocol's fault-tolerance can be a non-constant function of $r$, and we prove, for each $r$, matching upper and lower bounds on the optimal ``recoverable fault-tolerance'' achievable by any SMR protocol. For example, for protocols that guarantee liveness against 1/3-bounded adversaries in the partially synchronous setting, a 5/9-bounded adversary can always cause one consistency violation but not two, and a 2/3-bounded adversary can always cause two consistency violations but not three.  Our positive results are achieved through a generic ``recovery procedure''  that can be grafted on to any accountable SMR protocol and restores consistency following a violation while rolling back only transactions that were finalized in the previous $2\Delta^*$ timesteps.
\end{abstract}

\begin{document}

\maketitle

\section{Introduction} 

\textcolor{black}{Distributed ledgers and other distributed services rely on multiple participants maintaining a shared record, even when some of them may fail or behave maliciously.} We consider protocols for the state machine replication (SMR) problem,
in which processes receive transactions from an environment and are
responsible for finalizing a common sequence of transactions. We focus
on the partially synchronous setting~\cite{DLS88}, in which message
delays are bounded by a known parameter $\Delta$ following an unknown
``global stabilization time'' GST (and unbounded until that point).

\vspace{0.2cm} 
The two most basic requirements of an SMR protocol are consistency,
meaning that no two processes should finalize incompatible sequences
of transactions (one should be a prefix of the other), and liveness,
which stipulates that valid transactions should eventually be
finalized (ideally, following GST, within an amount of time
proportional to~$\Delta$). Guaranteeing consistency and liveness
becomes impossible if too many of the processes are faulty (i.e.,
deviate from the intended behavior of a protocol). For the SMR problem
in partial synchrony, it is possible to simultaneously guarantee
consistency against $\alpha$-bounded adversaries (i.e., adversaries
that control less than an $\alpha$ fraction of the participants) and
liveness against $\beta$-bounded adversaries if and only if
$\alpha+2\beta \le 1$  \cite{momose2021multi,DLS88}.

\vspace{0.2cm} 
The focus of this paper is consistency violations---the type of
violation that enables, for example, double-spending a cryptocurrency
native to a blockchain protocol. What can be said about a protocol
when the adversary is large enough to cause a consistency violation?
For example, is it already in a position to cause an unbounded number
of consistency violations (as opposed to just one), or could the
protocol ``fight back'' in some way?

\vspace{0.2cm} 
To make sense of this question and the idea of multiple consistency
violations, we must formalize a sense in which a protocol might
restore consistency following a violation, necessarily by rolling back
transactions that had been viewed as finalized by some non-faulty
processes. One key parameter is then the {\em recovery time}~$d$,
meaning the number of timesteps after a violation before a protocol returns to healthy operation.  A second is the {\em
  rollback}~$h$, meaning that the recovery process ``unfinalizes''
only transactions that have been finalized within the previous $h$
time steps. \textcolor{black}{We say that a protocol has \emph{bounded
rollback} if such a bound $h$ can be fixed independently of the
particular execution (see Section~\ref{rmsec} for the formal definition).}

The natural wishlist for an SMR protocol in partial synchrony would
then be:
\begin{itemize}

\item [(1)] All of the ``usual'' guarantees, such as optimal fault-tolerance
  (i.e., consistency with respect to $\alpha$-bounded adversaries and
  liveness with respect to $\beta$-bounded adversaries for some
  $\alpha,\beta > 0$ with $\alpha+2\beta=1$).

\item [(2)] Automatic recovery from a consistency violation with the
  worst-case recovery time~$d$ and worst-case rollback~$h$ as small as
  possible (if nothing else, independent of the specific execution).

\item [(3)] Never suffers more than~$r$ consistency violations overall,
  where~$r$ is as small as possible.

\end{itemize}
To what extent are these properties simultaneously achievable?

\vspace{0.2cm} 
This paper provides a thorough investigation of this question.
\textcolor{black}{Since recovery from consistency violations with
bounded rollback is in general impossible in the pure partially
synchronous setting (Theorem~\ref{ir1}),} we work with a timing model that
can be viewed as an interpolation between the synchronous and
partially synchronous settings. In addition to the usual
parameters~$\Delta$ and GST (known and unknown, respectively) of the
partially synchronous model, we allow for a known
parameter~$\Delta^* \ge \Delta$ which may or may not bound message
delays prior to GST. Canonically, $\Delta^*$ should be thought of as
orders of magnitude larger than~$\Delta$, with $\Delta$ indicating the speed of
communication between processes when all is well (no network issues,
no attacks) and $\Delta^*$ a safe (and possibly large) upper-bound on the speed of (possibly out-of-band)
communication around the time of an 
attack.\footnote{Indeed, for our positive results, message delays must
  be bounded by~$\Delta^*$ \textcolor{black}{around the time of a consistency
  violation and throughout} our recovery procedure,
  but not otherwise.}
We will be interested  in protocols that always satisfy all the ``usual''
guarantees~(1) and that finalize transactions in time $O(\Delta)$ (rather than $O(\Delta^*)$) after GST, whether or not pre-GST message delays are bounded
by $\Delta^*$, and also satisfy the additional recovery guarantees~(2)
and~(3) in the event that pre-GST message delays are in fact bounded
by~$\Delta^*$.\footnote{In particular, a synchronous protocol with
  respect to the parameter~$\Delta^*$ will not generally satisfy
  consistency and liveness if message delays do not happen to \textcolor{black}{be} bounded
  above by $\Delta^*$.}

\vspace{0.2cm} 
Our main positive result, stated formally in Theorem~\ref{mt} and
proved in Section~\ref{pmt}, shows that such protocols do indeed
exist. For example, we show that there is a protocol that satisfies:
\begin{itemize}

\item $\tfrac{1}{3}$-resilience in partial synchrony (independent of
  whether $\Delta^*$ bounds pre-GST message delays), with expected latency $O(\Delta)$ after GST;  
  
\item should pre-GST message delays be bounded by~$\Delta^*$, recovers
  from consistency violations in expected time~$O(\Delta^*)$ and with
  rollback $2\Delta^*$; and

\item should pre-GST message delays be bounded by~$\Delta^*$, never
  suffers from more than one consistency violation with a
  $\tfrac{5}{9}$-bounded attacker, and never suffers from more than
  two consistency violations with a $\tfrac{2}{3}$-bounded attacker.

\end{itemize}
We achieve this result by designing a generic ``recovery procedure''
that can be grafted on to any accountable SMR protocol, including
protocols with asymmetric fault-tolerance with respect to consistency
and liveness attacks. \textcolor{black}{Informally, an SMR protocol
is \emph{accountable} if a consistency violation yields messages
constituting \emph{proofs of guilt} for faulty processes (see
Section~\ref{setup} for the formal definition). Several standard SMR
protocols have been shown to satisfy accountability, and generic
transformations can also provide it (see Related Work).}
Sections~\ref{intuition} and~\ref{formalspec}
give informal and formal, respectively, descriptions of this procedure.

\vspace{0.2cm} 
Our results are tight in several senses. For example, we prove in
Theorem~\ref{ir1} that recovery from a consistency violation
necessarily requires a rollback proportional to the
parameter~$\Delta^*$. In particular, in the pure partially synchronous
model ($\Delta^*=+\infty$, in effect), recovery from consistency
violations with bounded rollback is impossible. Theorems~\ref{ir3}
and~\ref{ir4} show that the bounds we obtain on adversary size (as a
function of the number~$r$ of consistency violations) are optimal. For
example, in the symmetric case above, an attacker controlling
five-ninths of the processes can always force two consistency
violations, and one controlling \textcolor{black}{more than} two-thirds of the processes can cause
unbounded rollback \textcolor{black}{if the number of consistency violations
is bounded independently of $n$}.

\vspace{0.2cm} 
\noindent \textbf{A comment on the benefits of automated recovery}. It
is sometimes assumed in the literature (e.g.\
\cite{sridhar2023better}) that, in the event of a consistency
violation, an \textcolor{black}{``administrator''} will somehow (through out-of-protocol
means) remove perpetrators from the system and co-ordinate an
appropriate \textcolor{black}{``reboot.''} \textcolor{black}{In a distributed ledger
securing substantial monetary or contractual value, however,
decisions about which parties to remove and which transactions to
roll back may themselves have significant economic consequences.
In particular, during an attack, there may be no external
administrator whose decisions all participants would regard as
legitimate.} Our automated recovery procedure is
``in-protocol,'' and therefore \textcolor{black}{does not rely on
such an external authority or create a centralized point of failure.
Rather, it provides a procedure for agreeing on which provably faulty
parties to remove and on a new genesis log, while guaranteeing
bounded rollback.}

\vspace{0.2cm} 
\noindent \textbf{A comment on the $(\Delta, \Delta^*)$ timing assumptions}. \textcolor{black}{In practical settings, partial synchrony may hold with respect to a relatively small bound $\Delta$, while it may also be reasonable to expect that message delays will satisfy a much larger bound $\Delta^*$ around the time of an attack. Our aim is to design protocols that retain all the usual guarantees of partial synchrony, with post-GST latency depending on $\Delta$ rather than $\Delta^*$, whether or not the larger bound holds, but which also provide additional recovery guarantees when it does hold. One could instead assume synchrony with respect to $\Delta^*$ throughout, allowing greater fault tolerance than is possible under partial synchrony. However, this would be a stronger timing assumption, and could also lead to substantially greater latency. Although synchronous protocols can sometimes finalize transactions at a speed determined by actual message delays, rather than the assumed upper bound, this cannot be guaranteed throughout the full range of faults tolerated by synchronous protocols with resilience greater than $1/3$.\footnote{This distinction is captured by the notions of \emph{responsiveness} and \emph{optimistic responsiveness}. Pass and Shi~\cite{pass2018thunderella} discuss these notions and the impossibility of responsiveness with at least one-third Byzantine faults. Optimistic responsiveness permits latency independent of the assumed delay bound under more favourable conditions; see, e.g., Sync HotStuff~\cite{abraham2020sync}. Abraham et al.~\cite{abraham2020goodcase} discuss the corresponding good-case latency lower bound.} More importantly, such protocols do not provide the combination of ordinary partial-synchrony guarantees and additional recovery guarantees that we seek here.}

\section{The setup} \label{setup} 

We consider a set $\Pi= \{ p_1,\dots, p_{n} \}$ of $n$ processes. Each
process $p_i$ is told its ``name'' $i$ as part of its input. 
We focus on the case of a static adversary, which chooses a set of
processes to corrupt at the start of the protocol
execution.\footnote{All results in this paper hold more generally for
  adaptive adversaries (with essentially identical proofs), with the
  exception of the bound on the
  expected termination time for the recovery procedure asserted in
part~(iii) of Theorem \ref{mt} (which requires that a random
permutation of the processes be chosen subsequent to the adversary
deciding which processes to corrupt).}
We call a process corrupted by the adversary 
\emph{faulty}. Faulty processes may behave arbitrarily (i.e., we consider
\emph{Byzantine} faults), subject to our cryptographic assumptions
(stated below).  Processes that are not faulty are \emph{correct}.
The adversary is \emph{$\rho$-bounded} if it corrupts less than a
$\rho$ fraction of the $n$ processes.

\vspace{0.2cm} 
\noindent \textbf{Cryptographic assumptions}. 
We assume that processes communicate by point-to-point authenticated
channels and that
a public key infrastructure (PKI) 
is available for generating and validating signatures.
For simplicity of presentation (e.g., to avoid the
analysis of negligible error probabilities), we work with ideal
versions of these primitives (i.e., we assume that faulty processes
cannot forge signatures).
We also assume that all processes have access to a random
permutation of $\Pi$, denoted $\Pi^*:\ [1,n]\rightarrow  \Pi$,  which
is sampled after the adversary chooses which processes to corrupt.

\vspace{0.2cm} 
\noindent \textbf{Message delays}. 
We consider a discrete sequence of
timeslots $t\in \mathbb{N}_{\geq 0}$.
As discussed in the introduction, we consider protocols that operate
in partial synchrony (with some parameter~$\Delta$, perhaps in the order of
seconds or milliseconds) and satisfy additional recovery properties
should synchrony hold (with some different parameter~$\Delta^*$, which may be set much larger than $\Delta$) while running the \textcolor{black}{``recovery procedure''}.

\vspace{0.1cm} 
\noindent \emph{Synchrony}. In the synchronous setting, a message sent at time $t$ must arrive by time $t + \Delta^*$, where $\Delta^*$ is known to the protocol. 

\vspace{0.1cm} 
\noindent \emph{Partial synchrony}. In the partially synchronous setting, a message sent at time $t$ must arrive by time $\max\{\text{GST},t\} + \Delta$. While $\Delta$ is known, the value of GST is unknown to the protocol. The adversary chooses GST and also message delivery times, subject to the constraints already specified.

\vspace{0.2cm} 
\noindent \textbf{Clock synchronization}. In the partially synchronous
setting, we suppose all correct processes begin the protocol execution
before GST. When considering the synchronous setting, we suppose all
correct processes begin the protocol execution by time $\Delta^*$. A
correct process begins the protocol execution with its local clock set
to 0; thus, we do not suppose that the clocks of correct processes are
synchronized. For simplicity, we assume that
the clocks of correct processes all proceed in real time, 
meaning that if
$t'>t$ then the local clock of correct $p$ at time $t'$ is $t'-t$ in
advance of its value at time $t$.\footnote{Using standard arguments, our
  protocol and analysis can easily 
be extended to the case
in which there is a known upper bound on the difference
between the clock speeds of correct processes.}

\vspace{0.2cm} 
\noindent \textbf{Notation concerning executions and received messages}. We use the following notation \textcolor{black}{for an arbitrary execution of a protocol}: 
\begin{itemize} 
\item $M_i(t)$ denotes the set of messages received by process $p_i$ by timeslot $t$;
\item  $M_c(t)$  denotes the set of all messages received by any correct process by timeslot $t$; 
\item $M_c$  denotes the set of all messages received by any correct process during the execution. 
\end{itemize}

\vspace{0.2cm} 
\noindent \textbf{Transactions}. Transactions are messages of a
distinguished form,  signed by the \emph{environment}. Each timeslot, each process may receive some
finite set of transactions directly from the environment.

\vspace{0.2cm} 
\noindent \textbf{Determined inputs}. A  value is \emph{determined} if it \textcolor{black}{is} known to all processes, and is otherwise \emph{undetermined}. For example, $\Delta$, $\Delta^*$ and $\Pi$ are determined, while GST is undetermined.

\vspace{0.2cm} 
\noindent \textbf{State machine replication}. 
Informally, a protocol for state machine replication (SMR) must cause correct
processes to finalize \emph{logs} (sequences of transactions) that are
live and consistent with each other, and must also produce \textcolor{black}{``certificates''} of finalization that can be presented to clients who may not always be online and observing the protocol execution.  Formally, if $\sigma$ and $\tau$
are sequences, we write $\sigma \preceq \tau$ to denote that
$\sigma$ is a prefix of $\tau$.  We say $\sigma$ and~$\tau$ are
\emph{compatible} if $\sigma\preceq \tau$ or $\tau\preceq \sigma$.
If two sequences are not compatible, they are \emph{incompatible}.
If~$\sigma$ is a sequence of transactions, we write
$\mathtt{tr}\in \sigma$ to denote that the transaction $\mathtt{tr}$
belongs to the sequence $\sigma$.  

Fix a process set $\Pi$ and
\emph{genesis log}, denoted $\mathtt{log}_G$.  If $\mathcal{P}$ is a
protocol for SMR, then it must specify a function $\mathcal{F}$, which
may depend on $\Pi$ and $\mathtt{log}_G$, that maps any set \textcolor{black}{$M$} of
messages to a sequence of transactions extending $\mathtt{log}_G$.\textcolor{black}{\footnote{\textcolor{black}{Every transaction in $\mathcal{F}(M)$ must occur in $\mathtt{log}_G$ or in $M$, possibly within another message.}}} 
We
require the following conditions to hold in every execution (for any $M_1,M_2,p_i,p_j$ and $\mathtt{tr}$):

\vspace{0.1cm} 
\noindent \emph{Consistency}.   If $M_1\subseteq M_2 \subseteq M_c$,
then
$\mathcal{F}(M_1)\preceq \mathcal{F}(M_2)$.\footnote{This is
  equivalent to the seemingly stronger condition in which $M_c$ is
  replaced by the set of messages received by any process (correct or
  otherwise), as faulty processes always have the option of echoing
  any messages they receive to correct processes.}

\vspace{0.1cm} 
\noindent \emph{Liveness}. If $p_i$ and $p_j$ are correct and if $p_i$ receives the transaction $\mathtt{tr}$ then, for  some  $t$,  $\mathtt{tr}\in \mathcal{F}(M_j(t))$. 

\vspace{0.2cm} This definition of consistency ensures that correct
processes never finalize incompatible logs: for any sets
$M_1,M_2\subseteq M_c$ that two such processes might have received,
$ \mathcal{F}(M_1) \preceq \mathcal{F}(M_1 \cup M_2)$ and
$\mathcal{F}(M_2)\preceq \mathcal{F}(M_1\cup M_2)$.  We say a set of
messages $M$ is a \emph{certificate} for a sequence of transactions
$\sigma$ if $\mathcal{F}(M)\succeq \sigma$: intuitively, any  process can present the set of (potentially signed) messages $M$ to a \textcolor{black}{``client''} that has not been observing the protocol execution as proof that it has finalized a sequence of transactions extending $\sigma$.  If we wish to make
$\mathcal{F}$ explicit, we may also say that $M$ is an
$\mathcal{F}$-certificate for $\sigma$.

\vspace{0.2cm} 
\noindent \textbf{SMR (informal discussion)}. The
 selection~$\mathcal{F}$ of finalized transactions by a correct
 process depends only on the set of messages it has received, and not
 on the times at which these messages were received. The motivation
 for this restriction is to formalize the requirement of SMR \cite{schneider1990implementing} that \textcolor{black}{``clients''} wishing to verify
 the finality of transactions without
 observing the entire execution of the protocol should be able to do so  (via a suitable certificate). As discussed  in  \cite{sridhar2024consensus}, this crucial distinction between SMR and Total-Order-Broadcast is sometimes overlooked in the literature:  In the context of an honest majority, any report of finality by a majority of processes constitutes proof of finality, so a protocol for Total-Order-Broadcast directly gives a protocol for SMR. However, the distinction becomes non-trivial in the context of player reconfiguration, or if there is no guarantee of honest majority (such as in this paper). In
 partial synchrony, certificates are \textcolor{black}{in any case} required for guaranteed
  consistency and liveness~\cite{lewis2021does}, so our approach is general for the context considered here.

\vspace{0.2cm} 
\noindent \textbf{The liveness parameter}. If there exists some fixed $\ell$ that is a function of determined inputs\footnote{The requirement that $\ell$ is not a function of $\Delta^*$ (while $\Delta^*$ is not necessarily $O(\Delta)$) means that having liveness parameter $\ell$ may require finalization of transactions in time less than $\Delta^*$ after GST.} other than $\Delta^*$ and such that the following holds in all executions of $\mathcal{P}$, we say $\mathcal{P}$ has \emph{liveness parameter} $\ell$: If $p_i$ and $p_j$ are correct and if $p_i$ receives the transaction $\mathtt{tr}$ at time $t$ then, for  $t'=\text{max}\{ t,\text{GST} \} +\ell$,  $\mathtt{tr}\in \mathcal{F}(M_j(t'))$.

\vspace{0.2cm} 
\noindent \textbf{Liveness and consistency resilience}. Recall that
$n=|\Pi|$. When the protocol $\mathcal{P}$ is clear from context, we
write $\rho_C$ to denote the consistency resilience of $\mathcal{P}$,
which is the largest $\rho$ such that, for all $n$, the protocol
satisfies consistency so long as the adversary is $\rho$-bounded.  We
write $\rho_L$ to denote the liveness resilience, which is the largest
$\rho$ such that, for all $n$, the protocol satisfies liveness so long
as the adversary is $\rho$-bounded.  It is well known that
$\rho_C+2\rho_L\leq 1$ in the partially synchronous setting \cite{DLS88} and that
$\rho_C+\rho_L\leq 1$ in the synchronous setting \cite{momose2021multi}.

\vspace{0.2cm} 
\noindent \textbf{The number of consistency violations}. 
When $\mathcal{F}$ is clear from context, we say the set of messages
$M$ \emph{has $r$ consistency violations} if there exist \textcolor{black}{finite sets}
$M_0 \subset M_1 \subset \dots \subset M_r\subseteq M$ such that, for
each $s\in \{0,1,\ldots,r-1\}$,
$\mathcal{F}(M_s) \not \preceq \mathcal{F}(M_{s+1})$. We also say
$M$ \emph{has a consistency violation} (w.r.t.\ $\mathcal{F}$) if it
has at least one consistency violation. An execution has $r$
consistency violations if $M_c$ has $r$ consistency violations.

\vspace{0.2cm} 
\noindent \textbf{Accountable protocols (informal
  discussion)}. Informally, a protocol is \emph{accountable} if it
produces \emph{proofs of guilt} for some faulty processes in the event
of a consistency violation. We cannot generally require proofs of
guilt for a fraction $\lambda>\rho_C$ of processes, since consistency
violations may occur when less than a fraction $\lambda$ of processes
are faulty. On the other hand, all standard protocols that provide
accountability produce proofs of guilt for a $\rho_C$ fraction of
processes in the event of a consistency violation \cite{sheng2021bft}.

\vspace{0.2cm} 
\noindent \textbf{Accountable protocols (formal definition)}. Consider an SMR protocol $\mathcal{P}$: 
\begin{itemize} 
\item We say the set of messages $M$ is a \emph{proof of guilt} for
  $p\in \Pi$ if there does not exist any execution of $\mathcal{P}$ in
  which $p$ is correct and for which $M\subseteq M_c$.\footnote{If we
    wish to make $\mathcal{P}$, $\mathtt{log}_{G}$, and $\Pi$
    explicit, we may also say that $M$ is a
    $(\mathcal{P},\Pi,\mathtt{log}_G)$-proof of guilt.}

\item We say $\mathcal{P}$ is $\lambda$-\emph{accountable} if the
  following holds at every timeslot $t$ of any execution of
  $\mathcal{P}$: if $M_c(t)$ has a consistency violation, then
  $M_c(t)$ is a proof of guilt for at least a $\lambda$ fraction of
  processes in $\Pi$.
\end{itemize} 

\vspace{0.2cm} Given that all standard protocols that are
$\lambda$-accountable for any $\lambda>0$ are also
$\rho_C$-accountable, we will say that a protocol is
\emph{accountable} to mean that it is $\rho_C$-accountable.\footnote{\textcolor{black}{Our wrapper assumes proofs of guilt can be effectively identified and verified, as in standard accountable SMR protocols.}}  It is
important to note that, while an accountable protocol ensures the
existence of proofs of guilt for a $\rho_C$ fraction of processes in
the event of a consistency violation, it \emph{does not} automatically
ensure \emph{consensus} between correct processes as to a set of
faulty processes for which a proof of guilt exists. One role of the
\emph{recovery} procedure (as specified in Section \ref{formalspec})
will be to ensure such consensus.

\vspace{0.2cm} 
\noindent \textbf{Message gossiping}. In our recovery procedure, it
will be convenient to assume that correct processes gossip all
messages received.  Then, if synchrony does hold with respect to
$\Delta^*$, any message received by correct $p$ at some timeslot $t$
is received by all correct processes by time $t+\Delta^*$.  It will
not generally be necessary to gossip all messages; for example, for
standard quorum-based protocols, it will suffice to gossip blocks that
have received quorum certificates (QCs) along with those QCs.

\vspace{0.2cm} 
\noindent \textbf{A comment on setup assumptions}. Given an
accountable SMR protocol $\mathcal{P}$ and a process set $\Pi$, our
wrapper will initiate a sequence of executions of $\mathcal{P}$, 
with process sets that are progressively smaller subsets of $\Pi$.  Of
course, a PKI for $\Pi$ suffices to provide a PKI for each subset of
$\Pi$ and a random permutation of $\Pi$ naturally induces a random
permutation of each subset.

\section{Recovery metrics}  \label{rmsec} 

In this section, we introduce definitions to quantify how well a
protocol recovers from consistency violations.

\vspace{0.2cm} 
\noindent \textbf{Generalizing resilience to take recovery into
  account}. 
Is a protocol vulnerable to one consistency violation
inexorably doomed to an unbounded number of them? Or could a protocol
achieve strictly higher levels of resilience by tolerating (and recovering
from) a bounded number of consistency violations?
The following definitions generalize consistency and liveness
resilience to account for the possibility of recovery from consistency
violations.
\begin{itemize} 
\item \emph{Recoverable consistency resilience}. Consider a function
  $g:\mathbb{N}_{\geq 0} \rightarrow [0,1]$. We say a protocol
  $\mathcal{P}$ has recoverable consistency resilience $g$ if the
  following holds for each $r\in \mathbb{N}_{\geq 0}$: $g(r)$ is the
  largest $\rho$ such that, for all $n$, 
provided the adversary is
  $\rho$-bounded,
executions of $\mathcal{P}$
  have at most $r$ consistency violations.
\item \emph{Recoverable liveness resilience}.  Consider a function
  $g:\mathbb{N}_{\geq 0} \rightarrow [0,1]$. We say a protocol
  $\mathcal{P}$ has recoverable liveness resilience $g$ if the
  following holds for each $r\in \mathbb{N}_{\geq 0}$: $g(r)$ is the
  largest~$\rho$ such that, for all $n$, provided the adversary is
  $\rho$-bounded, liveness holds in all executions with precisely $r$
  consistency violations.\footnote{Prior to the $r$th consistency
    violation, a sufficiently large adversary may still be in a position
 to  cause
    a liveness violation.}
\end{itemize} 

\vspace{0.2cm} Suppose $\mathcal{P}$ has consistency resilience
$\rho_C$ and recoverable consistency resilience $g$. Note that
$g(0)=\rho_C$.  Also, $g$ is nondecreasing (i.e., $g(s)\geq g(r)$ for
$s>r$): if executions of $\mathcal{P}$ have at most~$r$ consistency
violations when the adversary is $\rho$-bounded, then this is also
true of all $s>r$. If $g(r+1)>g(r)$, the protocol effectively has
increased consistency resilience after $r$ consistency violations.

\vspace{0.2cm} 
\noindent \textbf{Recoverable resilience for our wrapper}. Suppose
$\mathcal{P}$ is accountable and has consistency resilience~$\rho_C$
and liveness resilience $\rho_L$ for partial synchrony with
$\rho_C+2\rho_L=1$ \textcolor{black}{and $0<\rho_L\leq\rho_C$}.
If we
identify some fraction $x$ of the processes in $\Pi$ as faulty and
then run an execution of $\mathcal{P}$ using the remaining processes,
there will be no consistency violation so long as less than a
fraction $x+\rho_C(1-x)$ of the processes in $\Pi$ are faulty. Given
this, let us define a sequence $\{ x_r \}_{r\in \mathbb{N}_{\geq 0}}$
by recursion:
\[ x_0=0, \hspace{0.3cm} x_{r+1}=x_r+\rho_C(1-x_r). \] 
\noindent Define: 
\begin{equation}\label{eq:g}
 g_1(r)=\text{min} \{ x_{r+1}, 1-\rho_L \}, \hspace{0.3cm}
 g_2(r)=\text{min} \{ x_r+\rho_L(1-x_r),1-\rho_L \}. 
\end{equation}
\noindent Given $\mathcal{P}$ as input, our wrapper produces an SMR
protocol with recoverable consistency resilience~$g_1$ and recoverable
liveness resilience $g_2$ as in~\eqref{eq:g}. For example, if 
$\rho_C=\rho_L=\tfrac{1}{3}$, then $g_1(0) = g_2(0) = \tfrac{1}{3}$, $g_1(1) =
g_2(1) = \tfrac{5}{9}$,  and $g_1(r) = g_2(r) = \tfrac{2}{3}$ for
all~$r \ge 2$.

\vspace{0.2cm} 
\noindent \textbf{Specifying the recovery time}. Next, we provide a
definition that captures the time required by a protocol to recover from
consistency violations. Suppose $\mathcal{P}$ has recoverable liveness
resilience $g$. We say $\mathcal{P}$ has \emph{recovery time $d$ with
  liveness parameter $\ell$ \textcolor{black}{ with respect to $g$}} if the following holds for all executions
$\mathcal{E}$ of~$\mathcal{P}$:

\begin{enumerate} 
\item[$(\dagger_{d,\ell})$] \textcolor{black}{Suppose $r\in\mathbb{N}_{\geq 0}$, $\mathcal{E}$ has precisely $r$ consistency violations, the adversary is $g(r)$-bounded, and $t$ is the least timeslot such that $M_c(t)$ has $r$ consistency violations ($t=0$ if $r=0$).} If correct $p_i$ receives the transaction $\mathtt{tr}$ at any timeslot $t'$ then, for every correct $p_j$ and for  $t''=\text{max}\{ t+d,\text{GST},t' \} +\ell$,  $\mathtt{tr}\in \mathtt{log}_{j}(t'')$.
\end{enumerate} 
\noindent In the above, $d$ should be thought of as a \textcolor{black}{``grace period''}
after consistency violations, after which liveness with parameter
$\ell$ must hold. In our construction, $d$ is governed by the length
of time it takes to run our recovery procedure.

\vspace{0.2cm} 
\noindent \textbf{Probabilistic recovery time}. Our recovery procedure
uses the random permutation $\Pi^*$---chosen after the adversary
chooses which processes to corrupt---to select ``leaders,'' and as such
it will run for a random duration.
To analyze this, we allow the grace period parameter~$d$ in
the definition above to depend on an error probability
$\varepsilon \in [0,1]$ and sometimes write~$d_{\varepsilon}$ to
emphasize this dependence.
We then make the following definitions:
\begin{itemize} 
\item We say that $(\dagger_{d,\ell})$ \emph{is ensured with
    probability at least $p$} if, for every choice of corrupted
  processes (consistent with a static $\rho$-bounded adversary), with
  probability at least~$p$ over the choice of~$\Pi^*$ (sampled from
  the uniform distribution), $(\dagger_{d,\ell})$ holds in every
  execution consistent with these choices and with the setting.

\item We say that $\mathcal{P}$ has \emph{\textcolor{black}{probabilistic} recovery time $d_{\varepsilon}$ with liveness parameter $\ell$ \textcolor{black}{ with respect to $g$}} if it holds for every $\varepsilon \in [0,1]$ that $(\dagger_{d_{\varepsilon},\ell})$ is ensured with probability at least $1-\varepsilon$. 
\end{itemize}

\vspace{0.2cm} 
\noindent \textbf{Recovery time for our wrapper}. Given $\mathcal{P}$
with liveness parameter $\ell$ as input, our wrapper will produce an
SMR protocol with (worst-case) recovery time $O(\Delta^* \cdot f_a)$,
\textcolor{black}{probabilistic} recovery time
$O(\Delta^* \cdot \log \tfrac{1}{\varepsilon})$, and liveness
parameter $\ell$, where $f_a$ denotes the actual (undetermined) number of
faulty processes.

\vspace{0.2cm}
\noindent \textbf{Bounding rollback}. We say that a protocol has \emph{rollback bounded by $h$} if the following holds for every execution consistent with the setting and  every correct $p_i,p_j\in \Pi$: if there exists an interval $I=[t, t+h]$  such that $\sigma \preceq \mathtt{log}_i(t')$ for all $t'\in I$, then $\sigma \preceq \mathtt{log}_j(t')$ for all sufficiently large $t'$. That is, consistency violations can ``unfinalize'' only transactions
that have been finalized recently, within the previous $h$
time steps. Here, $h$ can be any value that depends only on determined inputs.

\vspace{0.2cm} 
\noindent \textbf{Bounding rollback for our wrapper}. Given an SMR
protocol $\mathcal{P}$ with liveness resilience $\rho_L$ as input, our
wrapper will produce an SMR protocol with rollback bounded by
$h=2\Delta^*$ so long as synchrony holds for $\Delta^*$ and the
adversary is $(1-\rho_L)$-bounded. In fact, 
while the recovery procedure described in Section
\ref{formalspec} requires a common choice for $\Delta^*$, rollback
can be bounded on an individual basis, with each correct process
making their own personal choice of message delay bound
$\leq \Delta^*$. \textcolor{black}{Rollback} will be bounded by twice their personal
choice of bound, so long as that bound on message delay
holds.\footnote{The requirement that the choice be
  $\leq \Delta^*$ stems from the fact that the recovery procedure
  requires delays to be bounded by $\Delta^*$ to function correctly.}

\section{The intuition behind the wrapper} \label{intuition}

We describe a wrapper, which takes an accountable and optimally
resilient SMR protocol $\mathcal{P}$ as input, and which runs an
execution of $\mathcal{P}$ until a consistency violation
occurs.\footnote{By ``optimally resilient,'' we mean that the
  protocol's consistency resilience $\rho_C$ and liveness
  resilience~\textcolor{black}{$\rho_L$} in partial synchrony are both positive and
  satisfy $\rho_C+2\rho_L=1$ (as is the case for all of the ``usual''
  SMR protocols designed for the partially synchronous setting). This
  assumption is merely to simplify the presentation. For a
  non-optimally resilient protocol, the ``$1-\rho_L$'' term
  in~\eqref{eq:g} should be replaced by ``$(1+\rho_C)/2$''.}  Once
this happens, the wrapper triggers a \textcolor{black}{``recovery procedure''}, which
achieves consensus on a set of faulty processes $F$ for which a proof
of guilt exists, together with a long initial segment of the log
produced by $\mathcal{P}$ below which no consistency violation has
occurred. The wrapper then initiates another execution of
$\mathcal{P}$ that takes this log as its genesis log, with the processes in $F$ removed from the process set. This next execution is run until
another consistency violation occurs, and so on.

\vspace{0.2cm} 
\noindent \textbf{Specifying $\mathtt{log}_i(t)$ and
  $\mathcal{F}$}. While the formal definition of SMR in Section
\ref{setup} requires us to specify the finalization rule $\mathcal{F}$
(from which the transactions $\mathtt{log}_i(t)$ finalized by~$p_i$
can then be defined as $\mathcal{F}(M_i(t))$), it will be more natural
when defining our wrapper to specify $\mathtt{log}_i(t)$ directly, and
then later to define $\mathcal{F}$ such that
$\mathtt{log}_i(t)=\mathcal{F}(M_i(t))$. Recall that the given
protocol $\mathcal{P}$ satisfies consistency and liveness with respect
to a function that may depend on the process set and the genesis
log. We write $\mathcal{F}(\Pi,\mathtt{log}_G)$ to denote this
function.

\vspace{0.2cm}

In what follows, we use the variable $\mathcal{E}$ to denote an
execution of the wrapper (with process set $\Pi$ and $\mathtt{log}_G$
as the genesis log), which initiates successive executions
$\mathcal{E}_1, \mathcal{E}_2,\dots$ of $\mathcal{P}$, where
$\mathcal{E}_r$ has process set $\Pi_r$ and $\mathtt{log}_{G_r}$ as
the genesis log.  Process $p_i$ maintains local variables
$\mathtt{M}_i$ and $\mathtt{M}_{i,r}$ for each $r\geq 1$.\footnote{We
  use $\mathtt{M}_i$ when specifying the pseudocode, rather than
  $M_i(t)$, since $p_i$ only has access to its local clock and does
  not know the \textcolor{black}{``global''} value of $t$.}  The former records all
messages so far received in execution $\mathcal{E}$, while the latter
records all messages so far received in execution $\mathcal{E}_r$. We
suppose messages have tags identifying the execution in which they are
sent, and that $\mathtt{M}_{i,r} \subseteq \mathtt{M}_i$ at every
timeslot, for all correct $p_i$ and all $r$.

\vspace{0.2cm}
{\color{black}
\noindent \textbf{The structure of this section}. The wrapper has three
closely related tasks, which we describe in turn:
\begin{enumerate}
\item While running $\mathcal{E}_r$, correct processes finalize
transactions and distinguish ordinary finalization from
\emph{strong finalization}. This distinction allows us to bound
rollback even when a consistency violation occurs (Section~\ref{ebr}).
\item Once a consistency violation is detected, correct processes
exchange the information needed for recovery. The next leader
proposes a set $F$ of provably faulty processes to remove, together
with a new genesis log that preserves all strongly finalized logs
(Section~\ref{rpi}).
\item The remaining processes must agree on one proposal. The
recovery procedure uses views, votes, locks, and quorum certificates
to reach agreement, after which the wrapper starts another execution
of $\mathcal{P}$ (Section~\ref{agreeintuition}).
\end{enumerate}
The second and third tasks are not completely independent: a process
checks a proposal using the genesis messages it has received, while
the set of processes eligible to vote depends on the proposed set $F$.
These two features will be important in the agreement procedure.
}

\subsection{Ensuring bounded rollback} \label{ebr}
{\color{black}
A correct process may have to roll back a finalized transaction when
recovery begins. The first task is to distinguish transactions that
can still be rolled back from those that must be preserved in the
next genesis log.
}
 
In what follows,
we write $\rho_C$ and $\rho_L$ to denote the consistency and liveness
resilience of $\mathcal{P}$.  Each process $p_i$ executing the wrapper
maintains a value $\mathtt{log}_{i}$. Suppose the currently running
execution of $\mathcal{P}$ is $\mathcal{E}_r$. To ensure rollback
bounded by $2\Delta^*$, $p_i$ proceeds as follows:
\begin{itemize} 
\item While running the execution $\mathcal{E}_r$ of $\mathcal{P}$, and when  $p_i$ finds that some subset of $\mathtt{M}_{i,r}$ is an $\mathcal{F}(\Pi_r,\mathtt{log}_{G_r})$-certificate for $\sigma$ properly extending $\mathtt{log}_{i}$, it will set $\mathtt{log}_{i}$ to extend $\sigma$.
\item  Process $p_i$ will only \emph{strongly finalize} $\sigma$, however, once $\mathtt{log}_{i}$ has extended $\sigma$ for an interval of length $2\Delta^*$.
\item  Upon finding that $ \mathtt{M}_{i,r}$ has a consistency violation w.r.t.\ $\mathcal{F}(\Pi_r,\mathtt{log}_{G_r})$, $p_i$ will: 
\begin{itemize} 
\item Send a signed $r$-\emph{genesis} message $(\text{gen},\mathtt{log}_{i},r)$ to all processes  (motivation below); 
 \item Temporarily set  $\mathtt{log}_{i}$ to be $\mathtt{log}_{G_r}$;  
 \item Stop running $\mathcal{E}_r$,  and;
 \item Begin the recovery procedure. 
\end{itemize} 
\end{itemize}

\vspace{0.2cm} To see what this achieves (modulo complications that
may later be introduced by the recovery procedure), suppose that
synchrony holds for $\Delta^*$. Then, due to our assumptions on
message gossiping described in Section \ref{setup}, $2\Delta^*$ bounds
the round-trip time between any two correct processes.  In particular,
suppose that $p_i$ finalizes $\sigma$ at $t$ because there exists
$M\subseteq \mathtt{M}_{i,r}$ which is an
$\mathcal{F}(\Pi_r,\mathtt{log}_{G_r})$-certificate for $\sigma$. Then
every correct process $p_j$ will receive the messages in $M$ by
$t+\Delta^*$, and will then finalize $\sigma$ (never to subsequently
finalize anything incompatible with $\sigma$), unless
$\mathtt{M}_{j,r}$ has a consistency violation (w.r.t.\
$\mathcal{F}(\Pi_r,\mathtt{log}_{G_r})$) by that time.  In the latter
case, $p_i$ will begin the recovery procedure by timeslot
$t+2\Delta^*$ and will not strongly finalize $\sigma$.

\vspace{0.2cm} 
\noindent \textbf{Complications introduced by the recovery
  procedure}. Our recovery procedure introduces the complication that
there is not necessarily consensus on which logs have been strongly
finalized by some correct process. If a single correct process has
strongly finalized $\sigma$ when the recovery procedure is triggered,
and if the procedure determines that a log
$\sigma'\not \succeq \sigma$ should be used as the genesis log in
the next execution of $\mathcal{P}$, then this may violate the
condition that the protocol has rollback bounded by $2\Delta^*$. We
must therefore ensure that the recovery procedure reaches consensus on
a log that extends all logs strongly finalized by correct
processes. As \textcolor{black}{explained} in Section \ref{rpi}, the $r$-genesis
messages sent by processes before entering the recovery procedure will
be used to achieve this.

\subsection{\textcolor{black}{Preparing a recovery proposal}} \label{rpi}
{\color{black}
We now describe how the recovery procedure obtains the two objects
that must be agreed upon: a set $F$ of provably faulty processes to
remove, and a sequence $\sigma$ to use as the next genesis log.
The first gives an honest majority among the remaining processes. 
The second ensures that recovery does not undo strongly finalized
transactions.
}

Recall that $\rho_C$ ($\rho_L$) is the consistency (liveness)
resilience of $\mathcal{P}$ in partial synchrony (with
parameter~$\Delta$), and that the wrapper aims to deliver extra
functionality in the case that synchrony happens to hold with respect
to the (possibly large) bound $\Delta^*$, and so long as the adversary
is $(1-\rho_L)$-bounded. So, suppose these conditions hold.
 
\vspace{0.2cm} As noted in Section \ref{ebr}, while running execution
$\mathcal{E}_r$ of $\mathcal{P}$, process $p_i$ will enter the
recovery procedure upon finding that $ \mathtt{M}_{i,r}$ has a
consistency violation. Given our gossip assumption, described in
Section \ref{setup}, this means that correct processes will begin the
recovery procedure within time $\Delta^*$ of each other.  The key
observation behind the recovery procedure is that, if one has a proof
of guilt for processes in some set $F$, where $|F|\geq \rho_Cn$, then
the fact that the adversary is $(1-\rho_L)$-bounded (and
$2\rho_L+\rho_C=1$) means that the adversary controls less than half
the processes in $\Pi-F$. This follows since:
\[ 
1-\rho_L-\rho_C=\rho_L,\text{ and so }(1-\rho_L-\rho_C)/(1-\rho_C)=
\textcolor{black}{\frac{\rho_L}{2\rho_L}=\frac12}.
\]
{\color{black}
This calculation concerns the first recovery, when $\Pi_r=\Pi$.
In later recoveries, the processes already removed are all faulty,
and removing further faulty processes cannot increase the fraction
of faults among those remaining.
}
\noindent As a consequence, we can run a modified version of a
standard ($\tfrac{1}{2}$-resilient) SMR protocol for synchrony (our protocol is most similar to  \cite{abraham2020sync}), in
which the instructions are divided into views, each with a distinct
leader. In each view, the leader makes a proposal for the set of
processes $F$ that should be removed from $\Pi_r$ to form $\Pi_{r+1}$,
and \emph{the processes outside $F$ then vote on that proposal}.

\vspace{0.2cm}
{\color{black}
\noindent \textbf{Choosing the next genesis log.} A set of proofs of
guilt identifies processes that may be removed, but does not
determine a safe genesis log. For this we use the $r$-genesis messages
sent when recovery begins.
}

As well as proposing $F$, the leader
$p_i$ must also suggest a sequence $\sigma$ to be used as
$\mathtt{log}_{G_{r+1}}$ and this sequence must extend all logs
strongly finalized by correct processes. To achieve this (while
keeping the probabilistic recovery time small), we run a short
sub-procedure at the beginning of the recovery procedure, before
leaders start proposing values.  We proceed as follows:
\begin{itemize} 
\item Each correct $p_j$ waits time $2\Delta^*$ upon beginning the
  recovery procedure and then sets $P_{j}(r)$ to be the set of
  processes in $\Pi_r$ from which it has received an $r$-genesis
  message.
\item Process $p_j$ then enters view $(r,1)$ (the $1^{\text{st}}$ view
  of the $r^{\text{th}}$ execution of the recovery procedure).
\end{itemize}

{\color{black}
Since correct processes begin recovery within $\Delta^*$ of each
other, each correct $p_j$ receives an $r$-genesis message from
every correct process before fixing $P_j(r)$. This will allow it
to check that the proposed new genesis log is safe.
}

\noindent To form an appropriate proposal $\sigma$ for
$\mathtt{log}_{G_{r+1}}$ while in view $(r,v)$, the leader $p_i$ of
the view waits for $2\Delta^*$ after entering the view (to accommodate
possible lags between the progress of and information received by
different correct processes), and then
proceeds as follows.  If $M$ is the set of $r$-genesis messages that
$p_i$ has received by that time and which are signed by processes in
$\Pi_r-F$, then let $M'$ be a maximal subset of $M$ that contains at
most one message signed by each process. We say $\sigma$ is extended
by the $r$-genesis message $(\text{gen},\sigma',r)$ if
$\sigma \preceq \sigma'$. Process $p_i$ then sets $\sigma$ so that
the following condition is satisfied:
  
\vspace{0.1cm}
\noindent $\dagger(M',\sigma)$: $\sigma$ is the longest sequence
extended by more than $\frac{1}{2}|\Pi_r-F|$ elements of $M'$.

\vspace{0.1cm} 
\noindent Process $p_i$ then sends $M'$ along with $\sigma$ as a
\emph{justification} for its proposal. A correct process $p_j$ will be
prepared to vote on the proposal if $\dagger(M',\sigma)$ is satisfied
and $M'$ includes messages from every member of $P_{j}(r)\textcolor{black}{-F}$. \textcolor{black}{(Alternatively, a proposal already certified in an earlier view may carry the corresponding quorum certificate (QC).)}

{\color{black}
To see why a correct leader can satisfy this condition, suppose
$p_j$ is correct and starts recovery at timeslot $t_j$. If
$p_k\in P_j(r)$, then $p_j$ has received an $r$-genesis message
signed by $p_k$ by $t_j+2\Delta^*$. By the message gossiping
convention in Section~\ref{setup}, every correct process receives
this message by $t_j+3\Delta^*$, even if $p_k$ is faulty.
A correct leader $p_i$, starting recovery at $t_i$, has
$t_j\leq t_i+\Delta^*$ and makes its proposal no earlier than
$t_i+4\Delta^*$. Thus $M'$ includes a message from every
member of $P_j(r)-F$, for every correct $p_j$.
}

 \vspace{0.2cm} To see that this achieves the desired outcome, note that if $p_j$ is correct and $M'$ includes messages from every member of $P_{j}(r)\textcolor{black}{-F}$, then it must contain a message from every correct process\textcolor{black}{, since $F$ contains only faulty processes}. If $\sigma'$ has been strongly finalized by some correct process, then every correct process must have finalized $\sigma'$ before entering the recovery procedure, and cannot have subsequently finalized any value incompatible with $\sigma'$. So, for each $r$-genesis message $(\text{gen},\sigma'',r)$ sent by a correct process, $\sigma''$ must extend $\sigma'$. It therefore holds that $\sigma' $ is extended by more than $\frac{1}{2}|\Pi_r-F|$ elements of $M'$, so that, if $\dagger(M',\sigma)$ is satisfied, $\sigma$ must extend $\sigma'$.

\subsection{\textcolor{black}{Agreement on a recovery proposal}} \label{agreeintuition}
{\color{black}
The preceding subsection describes how to prepare a proposal. We
still need all correct processes to agree on the \emph{same} proposal,
so that they remove the same faulty processes and use the same next
genesis log. This is the third task of the recovery procedure. To establish agreement, we combine the methods described in the preceding subsections with an adaptation of Sync HotStuff~\cite{abraham2020sync}.

\vspace{0.1cm}
\noindent \textbf{Views, votes and QCs}. For each $r$, the $r^{\text{th}}$ execution of the recovery
procedure is divided into successive views $(r,1),(r,2),\dots$,
each with a leader. A proposal for a view
is of the form $R=(P,v,Q)_{p_i}$, where $P=(F,\sigma,M,r)$ specifies
the processes to remove, the new genesis log, and the genesis messages $M$
that justify it. A new proposal has $Q=\bot$; alternatively, $Q$ may
be a quorum certificate (QC) for the same $P$ from an earlier view.
The leader first checks for such a QC and, if one exists, proposes
its value rather than preparing a new one.

Only votes from processes in $\Pi_r-F$ count for this proposal. A QC contains
votes from more than half of those processes. This set has an honest
majority, by the calculations of Section \ref{rpi}, and since $F$ consists of provably faulty processes. However,
different proposals may remove different sets of processes.

\vspace{0.1cm}
\noindent \textbf{Locks and finish votes}. The agreement procedure
uses the following instructions:
\begin{enumerate}
\item A correct process votes for a valid proposal at most once in
each view. For a new proposal, the validity check includes the
genesis-message condition described in Section \ref{rpi}. A proposal carrying a
valid QC from an earlier view is instead subject to the corresponding
QC and lock checks.
\item When a correct process receives a valid QC for the current
view, it sets its \emph{lock} to that QC and waits $2\Delta^*$.
It then sends a finish vote only if it has not detected equivocation by
the leader in that view by that time.
\item More than half of $\Pi_r-F$ signing finish votes for the same
$r$-proposal $P$ constitute a \emph{finish-QC}. Upon receiving a
valid finish-QC, a correct process can start the next execution of
$\mathcal{P}$ using the values specified by $P$.
\end{enumerate}
The delay before sending finish votes is important: if a correct
process can finish a view, synchrony ensures that other correct
processes have received the relevant QC and have had an opportunity
to detect any equivocation in that view. The locks then constrain
which proposals can obtain votes in later views.
}

\par\pagebreak[3]
\vspace{0.2cm}
\noindent \textbf{\textcolor{black}{Why the agreed proposal is unique.}}
To establish that at most one $r$-proposal can
receive a valid finish-QC, suppose that some correct $p_i$ sends a
finish vote for the $r$-proposal $P$ during view $(r,v)$. In this
case, $p_i$ must set its lock to some valid QC, $Q$ say, at some
timeslot $t$ while in view $v$. Suppose that $Q$ is a QC for the
$(r,v)$-proposal $R$, and note that $\mathtt{P}(Q)=P$. We will observe
that:
\begin{enumerate} 
\item All correct processes set their locks to some valid QC for $R$ while in view $v$.
\item No $(r,v)$-proposal other than $R$ can receive a QC that is regarded as valid by any correct process. 
\end{enumerate} 
\noindent From (1) and (2) it will be easy to argue by induction on
$v'>v$ that no correct process votes for any proposal
$R'=(P',v',Q')_{p_j}$ such that $P'\neq P$, since their locks will
forever prevent voting for such proposals. It follows that if any
correct $p_k$ sends a finish vote for an $r$-proposal $P'$ during some
view $v'\geq v$, then $P=P'$. We conclude that, assuming~(1) and~(2),
at most one $r$-proposal can receive a valid finish-QC.

\vspace{0.2cm} 
To see that (1) holds, note that all correct processes will be in view $(r,v)$ at $t+\Delta^*$ and will have received $Q$ by this time. They will therefore set their lock to be some QC for $R$, unless they have already received a valid QC for some $(r,v)$-proposal $R'\neq R$. The latter case is not possible, since then $p_i$ would detect equivocation in view $(r,v)$ by $t+2\Delta^*$, and so would not send the finish vote for $P$. 

\vspace{0.2cm} 
To see that (2) holds, the argument is similar. All correct processes will be in view $(r,v)$ at $t+\Delta^*$ and will have received $R$ by this time. Item (ix) in the validity conditions for $(r,v)$-proposals prevents correct processes from voting for $(r,v)$-proposals  $R'\neq R$ at later timeslots, and correct processes cannot vote for such proposals at any timeslot $\leq t+\Delta^*$ because $p_i$ would detect equivocation in view $(r,v)$ in this case.

\vspace{0.2cm}
\noindent \textbf{\textcolor{black}{Why an agreement is reached.}}
Having established that at most one $r$-proposal can
receive a valid finish-QC, suppose now, towards a contradiction, that
no $r$-proposal ever receives a valid finish-QC. Let $v$ be the least
such that $\mathtt{lead}(r,v)$ is correct and let $i$ be such that
$p_i=\mathtt{lead}(r,v)$. Since $p_i$ waits $2\Delta^*$, until some
timeslot $t$ say, before disseminating an $(r,v)$-proposal
$R=(P,v,Q)_{p_i}$, it will have seen all locks held by correct
processes by this time, and will have received $r$-genesis messages
from all processes in any set $P_{j}(r)$ for correct $p_j$. At $t$,
$p_i$ will disseminate an $(r,v)$-proposal which all correct processes
regard as valid by timeslot $t+\Delta^*$. All correct processes will
therefore vote for the proposal by this time and will receive a valid
QC for the proposal by time $t+2\Delta^*$. All correct processes will
then set their locks. They will still be in view $(r,v)$ by time
$t+4\Delta^*$ (since they enter the view within time $\Delta^*$ of
each other) and will send $r$-finish votes for $P$ by this time.

{\color{black}
Under synchrony, the finish votes then give a valid finish-QC for $P$.
When a correct process receives this QC, it sets
$\Pi_{r+1}=\Pi_r-F$ and $\mathtt{log}_{G_{r+1}}=\sigma$ and starts
$\mathcal{E}_{r+1}$. The agreed genesis log extends every log strongly
finalized by a correct process before recovery. Transactions received
but not included in the new genesis log are supplied as inputs to the
new execution. If another consistency violation is detected, the
same procedure is repeated.

This completes the description of the three tasks. Section~\ref{formalspec}
provides the formal definitions and the instructions implementing them.
}

\section{The formal specification of the wrapper} \label{formalspec}

In what follows, we suppose that, when a correct process sends a
message to \textcolor{black}{``all processes''}, it regards that message as immediately
received by itself.  The pseudocode uses a number of inputs, local
variables, functions and procedures, detailed below.
\textcolor{black}{We group the definitions by their role in the three
tasks described in Section~\ref{intuition}. In each group, definitions
are introduced before they are needed wherever possible.}

\vspace{0.2cm}
\par\medskip\noindent\textcolor{black}{\rule{\linewidth}{0.35pt}}\par\noindent\textcolor{black}{\emph{The underlying executions and local state.}}
\textcolor{black}{We first specify the inputs, the successive executions of
$\mathcal{P}$, and the logs maintained by each process.}

\vspace{0.1cm}
\noindent \textbf{Inputs}. The wrapper takes as input an SMR protocol
$\mathcal{P}$, a process set $\Pi$, a random permutation $\Pi^*$ of
$\Pi$, a value $\mathtt{log}_G$, and message delay bounds $\Delta^*$
and $\Delta$. The consistency resilience $\rho_C$ of $\mathcal{P}$ is
also given as input. Recall that the given protocol $\mathcal{P}$
satisfies consistency and liveness with respect to a finalization
function that may depend on the process set $\Pi'$ and the value
$\mathtt{log}_G'$ for the genesis log.  (For example, signatures from
a certain fraction of the processes in~$\Pi'$ may be required for
transaction finalization.)  We write
$\mathcal{F}(\Pi',\mathtt{log}_G')$ to denote this function, and
suppose also that this function is known to the protocol.
 
 \vspace{0.1cm} 
 \noindent \textbf{Permutations and the variables $\Pi_r$}. Process $p_i$ maintains a variable $\Pi_r$ for each $r\in \mathbb{N}_{\geq 1}$. $\Pi_1$ is initially set to $\Pi$, while each $\Pi_r$ for $r>1$ is initially undefined.\footnote{We write $x\uparrow$ to indicate that the variable $x$ is undefined, and $x\downarrow$ to indicate that $x$ is defined.} Once $\Pi_r$ is defined, $\Pi_r^*$ is the permutation of $\Pi_r$ induced by $\Pi^*$.

\vspace{0.1cm} 
\noindent \textbf{Received messages and executions}. We let
$\mathtt{M}_i$ be a local variable that specifies the set of all
messages so far received by $p_i$.  The wrapper will also initiate
executions $\mathcal{E}_1,\mathcal{E}_2,\dots$ of $\mathcal{P}$: for
each $r\geq 1$, $\mathtt{M}_{i,r}$ specifies all messages so far
received by $p_i$ in execution $\mathcal{E}_r$.  We suppose that
messages have tags identifying the execution in which they are sent,
and that all messages received by $p_i$ in $\mathcal{E}_r$ are also
received by $p_i$ in the present execution of the wrapper, so that
$\mathtt{M}_{i,r} \subseteq \mathtt{M}_{i}$ for all $r$.

\vspace{0.2cm} 
\vspace{0.1cm} 
\noindent \textbf{The variables $\mathtt{log}_{G_{r}}$}.  Process
$p_i$ maintains a variable $\mathtt{log}_{G_{r}}$ for each
$r\in \mathbb{N}_{\geq 1}$. Initially, $\mathtt{log}_{G_{1}}$ is set
to $\mathtt{log}_G$, while each $\mathtt{log}_{G_{r}}$ for $r>1$ is
undefined. If the execution $\mathcal{E}_r$ of $\mathcal{P}$ is
initiated by the wrapper, then this will be an execution with
$\mathtt{log}_{G_r}$ as the genesis log and with process set $\Pi_r$.
\textcolor{black}{Whenever $p_i$ starts $\mathcal{E}_r$, it supplies to that
execution every transaction in $\mathtt{M}_i$ that is not in
$\mathtt{log}_{G_r}$ as an input. While $\mathcal{E}_r$ is running,
it supplies further transactions as inputs when they are received,
unless already in $\mathtt{log}_{G_r}$ or supplied to $\mathcal{E}_r$.}
  
   \vspace{0.1cm} 
  \noindent \textbf{Logs}. Process $p_i$ maintains two variables $\mathtt{log}_{i}$ and $\mathtt{log}^*_{i}$. The former should be thought of as the sequence of transactions that $p_i$ has finalized, while the latter is the sequence that $p_i$ has strongly finalized. 
  
      \vspace{0.1cm} 
 \par\medskip\noindent\textcolor{black}{\rule{\linewidth}{0.35pt}}\par\noindent\textcolor{black}{\emph{Information used to prepare a recovery proposal.}}
\textcolor{black}{These definitions specify the signed genesis messages
and the information used to choose the processes to remove and the
next genesis log.}

\vspace{0.1cm}
\noindent \textbf{Signatures}. We write $m_{p_i}$ to denote the message $m$ signed by $p_i$. 
  
   \vspace{0.1cm} 
  \noindent \textbf{$r$-\emph{genesis} messages}. An $r$-\emph{genesis} message is a message of the form $(\text{gen},\sigma,r)_{p_j}$, where $\sigma$ is a sequence of transactions and $p_j\in \Pi$. These are used during the $r^{\text{th}}$ execution of the recovery procedure to help reach consensus on an appropriate value for $\mathtt{log}_{G_{r+1}}$. We say $\sigma'$ is extended by the $r$-genesis message $(\text{gen},\sigma,r)_{p_j}$ if $\sigma' \preceq \sigma$.

   \vspace{0.1cm} 
 \noindent \textbf{The variables $P_{i}(r)$ and $t_0$}. Process $p_i$ maintains a local variable $P_{i}(r)$ for each $r\geq 1$, initially undefined. Upon halting execution $\mathcal{E}_r$ and entering the recovery procedure at timeslot $t$ (according to its local clock), $p_i$ will set $t_0:=t$, wait $2\Delta^*$, and then set $P_{i}(r)$ to be the set of processes in $\Pi_r$ from which it has received signed $r$-genesis messages.

       \vspace{0.1cm} 
 \noindent \textbf{$r$-proposals}. An $r$-\emph{proposal} is a tuple  $P=(F,\sigma,M,r)$, where $F\subset \Pi$, $\sigma$ is a sequence of transactions, $M$ is a set of $r$-genesis messages,  and $r\in \mathbb{N}_{\geq 1}$. The last entry  $r$ indicates that this is a proposal corresponding to the $r^{\text{th}}$ execution of the recovery procedure.  One should think of $F$ as a suggestion for $\Pi_r-\Pi_{r+1}$, while $\sigma$ is a suggestion for $\mathtt{log}_{G_{r+1}}$ and $M$ is a \emph{justification} for $\sigma$.

     \vspace{0.1cm} 
\par\medskip\noindent\textcolor{black}{\rule{\linewidth}{0.35pt}}\par\noindent\textcolor{black}{\emph{Agreement on a recovery proposal.}}
\textcolor{black}{The remaining definitions specify the views, votes,
certificates, and checks used to agree on one proposal.}

\vspace{0.1cm}
\noindent \textbf{Views and leaders}. Views are indexed by ordered
pairs and ordered lexicographically: one should think of view $(r,v)$
as the $v^{\text{th}}$ view in the $r^{\text{th}}$ execution of the
recovery procedure. For $r,v\in \mathbb{N}_{\geq 1}$, we set
$\mathtt{lead}(r,v)=p_i$, where $p_i=\Pi_r^*(v)$; this function is
used to specify the leader of each view.\footnote{We can write
  $p_i=\Pi_r^*(v)$ because the number of views in the $r^{\text{th}}$
  execution of the recovery procedure will be bounded by~$|\Pi_r|$.}
 
\vspace{0.1cm} 
 \noindent \textbf{The time for each view}. Each view is of length $8\Delta^*$. Having set $t_0$ upon halting execution $\mathcal{E}_r$ and  entering the recovery procedure, $p_i$ will start view $(r,v)$ (for $v\geq 1$) at time $t_0+2\Delta^*+8(v-1)\Delta^*$. 
 
 \vspace{0.1cm} 
 \noindent \textbf{$(r,v)$-proposals}. An $(r,v)$-proposal is a message of the form $R=(P,v, Q)_{p_j}$, where $P$ is an $r$-proposal, $p_j\in \Pi$, and either $Q=\bot$ or else $Q$ is a QC (as specified below) for some $(r,v')$-proposal with $v'<v$. 
 
\vspace{0.1cm} 
\noindent \textbf{Votes}. A \emph{vote} for the $(r,v)$-proposal
$R=(P,v,Q)_{p_j}$, where $P=(F,\sigma,M,r)$, is a message of the form
$V=R_{p_k}$, where $p_k\in \Pi$. We also say $V$ is a \emph{vote by
  $p_k$}. At timeslot $t$, $p_i$ will regard $V$ as \emph{valid} if it
is
signed by one of the processes that, from~$p_i$'s perspective, remains
in the active process set---i.e., if~$\Pi_r$ is defined and $p_k\in \Pi_r-F$.
  
      \vspace{0.1cm} 
 \noindent \textbf{QCs}. A  QC for an $(r,v)$-proposal $R=(P,v,Q')_{p_j}$, where $P=(F,\sigma,M,r)$, is a set $Q$ of  votes for $R$. At timeslot $t$, $p_i$ will regard  $Q$ as valid if every vote in $Q$ is valid and  $Q$ contains more than $\frac{1}{2}| \Pi_r-F|$ votes, each by a different process\textcolor{black}{, and conditions (i)--(iii) in the definition of valid $(r,v)$-proposals below hold for $P$}. If $Q$ is a QC for an $(r,v)$-proposal $R=(P,v,Q')_{p_j}$, we set $\mathtt{view}(Q)=(r,v)$ and $\mathtt{P}(Q)=P$, and we may also just refer to $Q$ as a QC. 
 
     \vspace{0.1cm} 
 \noindent \textbf{Locks}. Each process $p_i$ maintains a value $Q_{i}^+$, which is initially undefined. This  variable should be thought of as playing the same role as locks in Tendermint.   The variable $Q_{i}^+$ may be set to a valid QC for an $(r,v)$-proposal during view $(r,v)$. 
 
      \vspace{0.2cm} 
 \noindent \textbf{Detecting equivocation}. At timeslot $t$, we say $p_i$ \emph{detects equivocation in view $(r,v)$} if $\mathtt{M}_{i}$ contains at least two distinct  $(r,v)$-proposals signed by $\mathtt{lead}(r,v)$.\footnote{If  $\mathtt{M}_{i}$ contains a vote for an $(r,v)$-proposal, we consider it as also containing that $(r,v)$-proposal.} %

\vspace{0.1cm} 
\noindent \textbf{Valid $(r,v)$-proposals}.  Consider an $(r,v)$-proposal $R=(P,v,Q)_{p_j}$, where $P=(F,\sigma,M,r)$. At timeslot $t$ (according to $p_i$'s local clock), process $p_i$ will regard  $R$ as valid if: 
 \begin{enumerate} 
 \item[(i)] $\Pi_r$ and $\mathtt{log}_{G_r}$ are defined; 
  \item[(ii)]  $F\subset \Pi_r$, and $|F|\geq \rho_C|\Pi_r|$; 
 \item[(iii)] \textcolor{black}{$\mathtt{M}_{i,r}$} is a $(\mathcal{P},\Pi_r,\mathtt{log}_{G_r})$-proof of guilt for every process in $F$; 
 \item[(iv)]  $M$ is a set of $r$-genesis messages, each signed by a different process in $\Pi_r-F$; 
 \item[(v)] \textcolor{black}{If $Q=\bot$, then} for each $p_k\in P_{i}(r)\textcolor{black}{-F}$, there exists an $r$-genesis message signed by $p_k$ in $M$; 
 \item[(vi)]  $\sigma$ is the longest sequence  extended by more than $\frac{1}{2}|\Pi_r-F|$ elements of $M$;
 \item[(vii)] $p_j=\mathtt{lead}(r,v)$;
 \item[(viii)] \textcolor{black}{if $Q\neq\bot$, then $Q$ is a valid QC with $\mathtt{P}(Q)=P$; and if $Q_i^+$ is defined, then $Q$ is a valid QC with $\mathtt{view}(Q)\geq\mathtt{view}(Q_i^+)$};
 \item[(ix)] $p_i$ does not detect equivocation in view $(r,v)$.
 \end{enumerate} 
 
\vspace{0.1cm} 
\noindent \textbf{The local variables $\mathtt{voted}$ and $\mathtt{lockset}$}. For each $(r,v)$, $\mathtt{voted}(r,v)$ and $\mathtt{lockset}(r,v)$ are initially set to 0. These values are used to indicate whether $p_i$ has yet voted or set its lock during view $(r,v)$. 

\vspace{0.1cm} 
\noindent \textbf{$r$-finish votes and QCs}. An \emph{$r$-finish vote for $P$} is a message of the form $P_{p_j}$, where $P=(F,\sigma,M,r)$ is an $r$-proposal and $p_j\in \Pi$. At timeslot $t$, $p_i$ will regard the $r$-finish vote as valid if $\Pi_r$ is defined and $p_j\in \Pi_r-F$. A \emph{valid finish-QC} for $P$ is a set of more than  $\frac{1}{2}| \Pi_r-F|$ valid $r$-finish votes for $P$, each signed by a different process\textcolor{black}{, such that conditions (i)--(iii) in the definition of valid $(r,v)$-proposals hold for $P$}.  

\vspace{0.1cm} 
\noindent \textbf{The procedure $\mathtt{Makeproposal}$}. If $p_i=\mathtt{lead}(r,v)$, then it will run this procedure during view $(r,v)$. To carry out the procedure, $p_i$ checks to see whether there exists some greatest $v'<v$ such that it has received a valid QC, $Q$ say,  with $\mathtt{view}(Q)=(r,v')$. If so, then $p_i$ sends the $(r,v)$-proposal $R=(\mathtt{P}(Q),v,Q)_{p_i}$ to all processes. If not, then  it sets $F$ to be  the set of all processes $p_j\in \Pi_r$ such that \textcolor{black}{$\mathtt{M}_{i,r}$} is a $(\mathcal{P},\Pi_r,\mathtt{log}_{G_r})$-proof of guilt for $p_j$. 
Let  $M$ be the set of $r$-genesis messages that $p_i$ has received and which are signed by processes in $\Pi_r-F$, and let  $M'$ be a maximal subset of $M$ that contains at most one message signed by each process.  Process $p_i$ then sets $\sigma$ to be the longest sequence  extended by more than $\frac{1}{2}|\Pi_r-F|$ elements of $M'$ and sends to all processes the 
 $(r,v)$-proposal $R=(P,v,\bot)_{p_i}$, where $P=(F,\sigma,M',r)$.  

\vspace{0.1cm} 
\par\medskip\noindent\textcolor{black}{\rule{\linewidth}{0.35pt}}\par\noindent\textcolor{black}{\emph{Conventions for the complete wrapper.}}

\vspace{0.1cm}
\noindent \textbf{Message gossiping}.  We adopt the  message gossiping  conventions described in Section \ref{setup}. 

\vspace{0.1cm} 
\noindent \textbf{The function $\mathcal{F}$}. \textcolor{black}{The function $\mathcal{F}$ is used at the end of recovery; elsewhere the pseudocode specifies the local log directly.} We will show in Section \ref{pmt} that, at every  $t$, $\mathtt{log}_i=\mathcal{F}(\mathtt{M}_i)$ (where $\mathtt{log}_i$ and $\mathtt{M}_i$ are as locally defined for $p_i$ at $t$). The function $\mathcal{F}$ is specified in Algorithm 2.

 \begin{algorithm}[tbp]
\caption{The instructions for $p_i$}\label{alg1}
\begin{algorithmic}[1]

    \State \textbf{Local variables} 
    
   \State  $\mathtt{r}$, initially 1.    \Comment Number of executions of $\mathcal{P}$ initiated
   
   \State $\mathtt{rec}$, initially 0.   \Comment 1 if carrying out recovery
   
   \State $\mathtt{log}_{i}$, $\mathtt{log}^*_{i}$, initially set to $\mathtt{log}_G$ \Comment Finalized and strongly finalized transactions

   \State $\Pi_r$, $r\geq 1$. Initially, $\Pi_1=\Pi$, while $\Pi_r\uparrow $ for $r>1$. \Comment Process set for $\mathcal{E}_r$
   
   \State $\mathtt{log}_{G_r}$. Initially, $\mathtt{log}_{G_1}=\mathtt{log}_G$, while $\mathtt{log}_{G_r}\uparrow$ for $r>1$. \Comment Genesis log for $\mathcal{E}_r$ 

   \State $\mathtt{M}_i$, $\mathtt{M}_{i,r}$, initially empty.    \Comment As specified in Section \ref{formalspec}
   
   \State $Q_{i}^+$, initially undefined. \Comment The lock 
   
   \State $t_0$, initially undefined. \Comment Timeslot at start of recovery 
   
   \State $P_{i}(r)$, initially undefined.    \Comment A set of processes
   
   \State $\mathtt{voted}(r,v)$, $\mathtt{lockset}(r,v)$ ($r,v\geq 1$),  initially 0.  \Comment As specified in Section \ref{formalspec}

    \State \textbf{At timeslot $t$:}

    \Statex \textcolor{black}{\rule{0.12\linewidth}{0.3pt}\enspace\footnotesize\textbf{Run $\mathcal{E}_{\mathtt{r}}$ and maintain the logs}}
    \State  \hspace{0.3cm}  \textbf{If} $t=0$, start execution $\mathcal{E}_1$ of $\mathcal{P}$, with process set $\Pi_1$ and with $\mathtt{log}_{G_1}$ as genesis log; \label{start}

            \State  \hspace{0.3cm}  \textbf{If} $\mathtt{rec}=0$:  \label{rec0a} 
        
            \State  \hspace{0.6cm} \textbf{If} $\mathtt{M}_{i,\mathtt{r}}$ has a consistency violation w.r.t.\ $\mathcal{F}(\Pi_\mathtt{r},\mathtt{log}_{G_\mathtt{r}})$:
            
            \State \hspace{0.9cm} Send $(\text{gen},\mathtt{log}_{i},\mathtt{r})_{p_i}$ to all processes;   \Comment Send $\mathtt{r}$-genesis message
            
             \State \hspace{0.9cm} Set $\textcolor{black}{\mathtt{log}_{i}}:=\mathtt{log}_{G_{\mathtt{r}}} $;  Stop running $\mathcal{E}_{\mathtt{r}}$; Set $\mathtt{rec}:=1$;   \label{rec0aend}     \Comment Start recovery

                  \State  \hspace{0.3cm}  \textbf{If} $\mathtt{rec}=0$:  \label{rec0b} 
              
               \State  \hspace{0.6cm}  \textbf{If}  $\exists \sigma,M$ s.t.\ $\sigma \succ \mathtt{log}_{i}$ and $ M\subseteq \mathtt{M}_{i,\mathtt{r}}$  is an  $\mathcal{F}(\Pi_\mathtt{r},\mathtt{log}_{G_\mathtt{r}})$-certificate for $\sigma$; 
              
                   \State  \hspace{0.9cm}  Let $\sigma$ be the  longest such; Set $\mathtt{log}_{i}:=\sigma$; \Comment Extend log 
              
       \State  \hspace{0.6cm}  \textbf{If} there exists a longest $\sigma \succ \mathtt{log}_{i}^*$ s.t. $\mathtt{log}_{i}$ has extended $\sigma$ for time $2\Delta^*$:
       
         \State  \hspace{0.9cm}  Set $\mathtt{log}^*_{i}:=\sigma$;   \label{rec0bend}  \Comment Extend strongly finalized log

    \Statex \textcolor{black}{\rule{0.12\linewidth}{0.3pt}\enspace\footnotesize\textbf{Prepare for recovery}}
                \State  \hspace{0.3cm}  \textbf{If} $\mathtt{rec}=1$:

                  \State  \hspace{0.6cm} \textbf{If} $t_0\uparrow$, set $t_0:=t$;  \label{recinit} \Comment Set $t_0$ upon entering recovery  
                  
                  \State  \hspace{0.6cm} \textbf{If} $t=t_0+2\Delta^*$:   \Comment Set $P_{i}(\mathtt{r})$
                  
                  \State  \hspace{0.9cm} Set $P_{i}(\mathtt{r}):= \{ p_j\in \Pi_{\mathtt{r}}:\ \mathtt{M}_i$ contains an $\mathtt{r}$-genesis message signed by $p_j \}$; \label{recinitend}

    \Statex \textcolor{black}{\rule{0.12\linewidth}{0.3pt}\enspace\footnotesize\textbf{Agreement on a recovery proposal}}
                   \State  \hspace{0.6cm} \textbf{If} $t=t_0+4\Delta^*+8(v-1)\Delta^*$ (for some $v\in \mathbb{N}_{\geq 1}$) \textbf{and} $p_i=\mathtt{lead}(\mathtt{r},v)$:  \label{viewv} 
                   
                    \State  \hspace{0.9cm} $\mathtt{Makeproposal}$;    \Comment Leader makes new proposal $2\Delta^*$ after starting view

                        \State  \hspace{0.6cm} \textbf{If} $t\in [t_0+2\Delta^*+8(v-1)\Delta^*, t_0+2\Delta^* +8v\Delta^*)$ (for some $v\in \mathbb{N}_{\geq 1}$):
                        
                           \State  \hspace{0.9cm} \textbf{If} $\mathtt{voted}(\mathtt{r},v)=0$ \textbf{and} $\mathtt{M}_i$ contains a valid $(\mathtt{r},v)$-proposal $R$: 
                           
                            \State  \hspace{1.2cm} Send $R_{p_i}$ to all processes; Set $\mathtt{voted}(\mathtt{r},v):=1$;  \Comment Vote 
                            
                               \State  \hspace{0.9cm} \textbf{If} $\mathtt{lockset}(\mathtt{r},v)=0$ \textbf{and} $\mathtt{M}_i$ contains a valid QC for an $(\mathtt{r},v)$-proposal, $Q$ say:
                             
                               \State  \hspace{1.2cm} Set $Q_{i}^+:=Q$, $\mathtt{lockset}(\mathtt{r},v):=1$;    \Comment Set lock
                               
                               \State  \hspace{1.2cm}  Set the $(\mathtt{r},v)$-timer to expire in time $2\Delta^*$;

                                 \State  \hspace{0.9cm} \textbf{If} $(\mathtt{r},v)$-timer expires and  $p_i$ does not detect equivocation in view $(\mathtt{r},v)$: 
                                                                  
                                  \State  \hspace{1.2cm}  Send $\mathtt{P}(Q_{i}^+)_{p_i}$ to all processes;    \Comment Send finish vote \label{viewvend}

    \Statex \textcolor{black}{\rule{0.12\linewidth}{0.3pt}\enspace\footnotesize\textbf{Adopt the proposal and restart $\mathcal{P}$}}
                                   \State  \hspace{0.6cm} \textbf{If} $\mathtt{M}_i$ contains a valid finish-QC for some $P=(F,\sigma,M,\mathtt{r})$: \label{finish1} 
                                   
                                     \State  \hspace{0.9cm}  Set $\Pi_{\mathtt{r}+1}:=\Pi_{\mathtt{r}}-F$, $\mathtt{log}_{G_{\mathtt{r}+1}}:=\sigma$;  \Comment Start new execution of $\mathcal{P}$
                                   
                                   \State  \hspace{0.9cm}  Set $\mathtt{r}:=\mathtt{r}+1$ and make $t_0$ and $Q_i^+$ undefined; 
                                   
                                      \State \hspace{0.9cm} Start execution $\mathcal{E}_{\mathtt{r}}$ of $\mathcal{P}$, with process set $\Pi_{\mathtt{r}}$ and with $\mathtt{log}_{G_{\mathtt{r}}}$ as genesis log;

                                     \State \hspace{0.9cm} Set $\mathtt{rec}:=0$; Set $\mathtt{log}_i:=\textcolor{black}{\mathcal{F}(\mathtt{M}_i)}$; \label{finish2}

\end{algorithmic}
\end{algorithm}

\par\medskip\noindent\textcolor{black}{\rule{\linewidth}{0.35pt}}\par
\noindent \textbf{Pseudocode walk-through}.
{\color{black}
\textcolor{black}{The groups of instructions in Algorithm~1 correspond to the tasks described
in Section~\ref{intuition}. We briefly explain their roles before giving the algorithm.
The check for a consistency violation precedes any further extension of the log.}

\par\smallskip\noindent
\textcolor{black}{\emph{Running the underlying protocol and detecting a violation.}}
\textbf{Line \ref{start}} starts $\mathcal{E}_1$, with process set
$\Pi_1=\Pi$ and genesis log $\mathtt{log}_{G_1}=\mathtt{log}_G$.
\textbf{Lines \ref{rec0a}--\ref{rec0aend}} check for a consistency violation
and, if one is found, send an $r$-genesis message, reset the log,
and begin recovery. \textbf{Lines \ref{rec0b}--\ref{rec0bend}} extend
the finalized and strongly finalized logs.

\par\smallskip\noindent
\textcolor{black}{\emph{Preparing for recovery.}}
\textbf{Lines \ref{recinit}--\ref{recinitend}} set $t_0$ and $P_i(r)$.
The latter set is used to check that a new proposal contains the
required genesis messages from processes in $P_i(r)-F$.

\par\smallskip\noindent
\textcolor{black}{\emph{Agreement on a recovery proposal.}}
\textbf{Lines \ref{viewv}--\ref{viewvend}} give the instructions for each view.
The leader makes a proposal, processes vote for a valid proposal, and
then set their locks upon receiving a valid QC. A locked process waits
$2\Delta^*$ and sends a finish vote if it has not detected equivocation.

\par\smallskip\noindent
\textcolor{black}{\emph{Adopting the proposal and restarting.}}
\textbf{Lines \ref{finish1}--\ref{finish2}} specify how a valid finish-QC
fixes $\Pi_{r+1}$ and $\mathtt{log}_{G_{r+1}}$, allowing the wrapper
to start the next execution of $\mathcal{P}$.

}

\begin{algorithm}[tbp]
\caption{The function $\mathcal{F}$}\label{alg2}
\begin{algorithmic}[1]

    \State \textbf{Inputs} 
    
    \State $M$ \Comment A set of messages 
    
    \State $\Pi$, $\mathtt{log}_G$ \Comment Process set  and genesis log

    \State $\mathcal{F}(\Pi',\mathtt{log}_G')$ \Comment A function for each possible $\Pi'$ and $\mathtt{log}_G'$
    
       \State \textbf{Local variables} 
    
   \State  $\mathtt{r}$, initially 1. 
   
   \State $\Pi_r$, $r\geq 1$. Initially, $\Pi_1=\Pi$, while $\Pi_r\uparrow $ for $r>1$. 
   
   \State $\mathtt{log}_{G_r}$. Initially, $\mathtt{log}_{G_1}=\mathtt{log}_G$, while $\mathtt{log}_{G_r}\uparrow$ for $r>1$. 
   
   \State $\mathtt{end}$, initially 0
   
   \State 
   
   \State \textbf{While} $\mathtt{end}=0$ \textbf{do}: 
   
   \State \hspace{0.3cm}  \textbf{If} $M$ does not have a consistency violation w.r.t.\ $\mathcal{F}(\Pi_\mathtt{r},\mathtt{log}_{G_\mathtt{r}})$:
   
             \State  \hspace{0.6cm}  Let $\sigma$ be longest such that $M$ is an  $\mathcal{F}(\Pi_\mathtt{r},\mathtt{log}_{G_\mathtt{r}})$-certificate for $\sigma$; 
              
                   \State  \hspace{0.6cm}  \textbf{Return} $\sigma$; Set $\mathtt{end}:=1$; 

   \State \hspace{0.3cm}  \textbf{Else if} there does not exist a unique $\mathtt{r}$-proposal with a valid finish-QC in $M$:  \label{sr} 
   
              \State  \hspace{0.6cm}  \textbf{Return} $\mathtt{log}_{G_{\mathtt{r}}}$; Set $\mathtt{end}:=1$; 
              
         \State \hspace{0.3cm}  \textbf{Else if} there exists a unique $\mathtt{r}$-proposal $P=(F,\sigma,M',\mathtt{r})$ with a valid finish-QC in $M$:   
         
           \State  \hspace{0.6cm}  Set $\Pi_{\mathtt{r}+1}:=\Pi_{\mathtt{r}}-F$, $\mathtt{log}_{G_{\mathtt{r}+1}}= \sigma$; 
           
               \State  \hspace{0.6cm} Set $\mathtt{r}:=\mathtt{r}+1$;

\end{algorithmic}
\end{algorithm}

\section{The theorem statements} 

Given functions $g,g':\mathbb{N}\rightarrow \mathbb{R}$, we say
$g\leq g'$ if $g(r)\leq g'(r)$ for all $r\in \mathbb{N}$. If
$x\in \mathbb{R}$, we say $g\leq x$ if $g(r)\leq x$ for all
$r\in \mathbb{N}$. We say $g< g'$ if $g\leq g'$ and $g(r)<g'(r)$ for
some \textcolor{black}{$r$}.

We begin with our main positive result, which states the guarantees
our wrapper achieves for recoverable consistency and liveness,
worst-case and probabilistic recovery time, and rollback.

\begin{theorem} \label{mt} Suppose the wrapper is given an accountable SMR protocol $\mathcal{P}$ as input, where $\mathcal{P}$ has consistency resilience $\rho_C$ and liveness resilience $\rho_L$  in partial synchrony, such that $\rho_C+2\rho_L=1$ \textcolor{black}{and $0<\rho_L\leq\rho_C$}. Let $g_1$ and $g_2$  be as
  defined in expression~\eqref{eq:g} in Section \ref{rmsec}.  If $\mathcal{P}$ has liveness parameter $\ell$ \textcolor{black}{uniformly across executions of $\mathcal{P}$ started by the wrapper} and is accountable for $(1-\rho_L)$-bounded adversaries, then the wrapper produces a protocol with the same consistency and liveness resilience as $\mathcal{P}$ in partial synchrony, and with the following properties for $(1-\rho_L)$-bounded adversaries when message delays are bounded by $\Delta^*$: 
\begin{itemize} 
\item[(i)] Recoverable consistency resilience $\geq g_1$ and recoverable liveness resilience $\geq g_2$. 
\item[(ii)]  Recovery time $O(f_a\Delta^*)$ with liveness parameter $\ell$ \textcolor{black}{ with respect to $g_2$}, where $f_a$ is the actual (unknown) number of faulty processes. 
\item[(iii)]   \label{rt} Probabilistic recovery time $O(\Delta^* \text{log}\frac{1}{\varepsilon} )$ with liveness parameter $\ell$ \textcolor{black}{ with respect to $g_2$}. 
\item[(iv)] Rollback bounded by $2\Delta^*$.
\end{itemize} 
\end{theorem}

\vspace{0.2cm} 
The proof of Theorem \ref{mt} is given in Section \ref{pmt}. 

The next three results describe senses in which Theorem \ref{mt} is
tight.
We say a protocol has \emph{bounded rollback} if there exists some $h$ such that  the protocol has rollback bounded by $h$.  
Our first impossibility result states that the rollback of a protocol
must scale with~$\Delta^*$, and hence bounded rollback in the
partially synchronous setting is impossible.

\begin{theorem}[Impossibility result 1] \label{ir1} 
Suppose partial synchrony holds w.r.t.\ $\Delta$ and synchrony holds w.r.t.\ $\Delta^*$.  
Suppose $\mathcal{P}$ is a protocol for SMR with liveness resilience $\rho_L$, consistency resilience $\rho_C\geq \rho_L$,  liveness parameter $\ell$, and with rollback bounded by $h$.  If we are given only that the adversary is $\rho$-bounded for $\rho>1-2\rho_L$, then $h=\Omega(\Delta^*)$. In particular, $\mathcal{P}$ does not have bounded rollback in the pure partially synchronous setting. 
\end{theorem} 
\begin{proof} 
The proof is an easy adaptation of the classic proof of Dwork, Lynch and Stockmeyer \cite{DLS88}. 
Let $\mathcal{F}$ be the function with respect to which  $\mathcal{P}$ has liveness resilience $\rho_L$ and consistency resilience $\rho_C$. 
Given any $\gamma\in [0,1)$, let $\Delta^*$ be large enough that $\textcolor{black}{(\ell+2)}/\Delta^*<1-\gamma$. We show that $\mathcal{P}$ does not have rollback bounded by $\gamma\Delta^*$.  Let $n$ be such that there exists $m\in \mathbb{N}$ with $1-2\rho_L<m/n<\rho$,  and such that $n-m$ is even. Let $\Pi=\{ p_1,\dots,p_n \}$ and let $\Pi_1,\Pi_2,\Pi_3$ be disjoint subsets of $\Pi$ with $|\Pi_1|=m$, $|\Pi_2|=|\Pi_3|=(n-m)/2$. Note that $(n-m)/2<\rho_Ln$. Let $\mathtt{tr}_1$ and $\mathtt{tr}_2$  be distinct transactions. \textcolor{black}{Take the genesis log to be empty, and suppose the environment issues no other transactions. We also use $\mathtt{tr}_1$ and $\mathtt{tr}_2$ for the corresponding one-transaction sequences.} We consider three executions of $\mathcal{P}$ with $\Delta=1$. 

\vspace{0.2cm} 
\noindent \textbf{Execution} $\mathcal{E}_1$. Processes in $\Pi_1$ and $\Pi_2$ are correct. Processes in $\Pi_3$ are faulty and perform no actions. Processes in $\Pi_1$ receive $\mathtt{tr}_1$ at timeslot 1. GST$=0$. 

\vspace{0.1cm} 
\noindent \textbf{Execution} $\mathcal{E}_2$. Processes in $\Pi_1$ and $\Pi_3$ are correct. Processes in $\Pi_2$ are faulty and perform no actions. Processes in $\Pi_1$ receive $\mathtt{tr}_2$ at timeslot 1. GST$=0$. 

\vspace{0.1cm} 
\noindent \textbf{Execution} $\mathcal{E}_3$. GST$=\Delta^*$. Processes in $\Pi_2$ and $\Pi_3$ are correct. Processes in $\Pi_1$ are faulty. Processes in $\Pi_1$ receive $\mathtt{tr}_1$ and $\mathtt{tr}_2$ at timeslot 1. The processes in $\Pi_1$ simulate two simultaneous executions $\mathcal{E}_3'$ and $\mathcal{E}_3''$ prior to GST, sending messages as instructed in these two executions, unless explicitly stated otherwise. In  $\mathcal{E}_3'$, each process in $\Pi_1$ acts exactly as if correct,  except that it (i)  ignores receipt of $\mathtt{tr}_2$, (ii)  ignores messages from, and does not send messages to, processes in $\Pi_3$, and (iii) ignores messages sent by processes in $\Pi_1$ corresponding to $\mathcal{E}_3''$. 
In  $\mathcal{E}_3''$, each process in $\Pi_1$ acts exactly as if correct,  except that it (i)  ignores receipt of $\mathtt{tr}_1$, (ii)  ignores messages from, and does not send messages to, processes in $\Pi_2$, and (iii) ignores messages sent by processes in $\Pi_1$ corresponding to $\mathcal{E}_3'$. 
 After GST, processes in $\Pi_1$ carry out no action.   
Message delays prior to GST are as follows: 
\begin{itemize} 
\item Messages sent by processes in $\Pi_2$ to processes in $\Pi_3$  do not arrive until GST.
\item Symmetrically, messages sent by processes in $\Pi_3$ to processes in $\Pi_2$ do not arrive until GST. 
\item For $k\in \{ 2,3 \}$, messages sent from processes in $\Pi_k$ to processes in $\Pi_1$ or $\Pi_k$  arrive at the next timeslot. 
\item Messages sent by processes in $\Pi_1$ to processes in $\Pi_1$, $\Pi_2$ or $\Pi_3$  arrive at the next timeslot. 
\end{itemize} 
Since $|\Pi_3|<\rho_L n$,  all processes in  $\Pi_2$ must finalize $\mathtt{tr}_1$ by timeslot $\textcolor{black}{1+\ell}$ in $\mathcal{E}_1$. Similarly,  all processes in  $\Pi_3$ must finalize $\mathtt{tr}_2$ by  timeslot $\textcolor{black}{1+\ell}$ in $\mathcal{E}_2$.  
Let $i$ and $j$ be such that $p_i\in \Pi_2$ and $p_j\in \Pi_3$. From  consistency (and since $\rho_C\geq \rho_L$) it follows that: 
\begin{enumerate} 
\item[(i)]  $\mathcal{F}(M_i(t))\textcolor{black}{\succeq}\mathtt{tr}_1$ for all $t\in [\textcolor{black}{1+\ell},\textcolor{black}{\Delta^*})$ in $\mathcal{E}_1$. 
\item[(ii)]  $\mathcal{F}(M_j(t))\textcolor{black}{\succeq}\mathtt{tr}_2$ for all $t\in [\textcolor{black}{1+\ell},\textcolor{black}{\Delta^*})$ in $\mathcal{E}_2$. 
\end{enumerate} 
For processes in $\Pi_2$,  $\mathcal{E}_3$ is indistinguishable from $\mathcal{E}_1$ prior to GST. For processes in $\Pi_3$,  $\mathcal{E}_3$ is indistinguishable from $\mathcal{E}_2$ prior to GST.  It follows that $\mathcal{F}(M_i(t))\textcolor{black}{\succeq}\mathtt{tr}_1$ for all $t\in [\textcolor{black}{1+\ell},\textcolor{black}{\Delta^*})$ in $\mathcal{E}_3$ and that  $\mathcal{F}(M_j(t))\textcolor{black}{\succeq}\mathtt{tr}_2$ for all $t\in [\textcolor{black}{1+\ell},\textcolor{black}{\Delta^*})$ in $\mathcal{E}_3$. However, it cannot be the case that $\mathcal{F}(M_k(t))$ begins with both $\mathtt{tr}_1$ and $\mathtt{tr}_2$ for all correct $p_k$ and all sufficiently large $t$. This suffices to establish the claim, since $\Delta^*-\textcolor{black}{(\ell+2)}>\gamma\Delta^*$. 
\end{proof}

Our second impossibility result justifies our restriction to
$(1-\rho_L)$-bounded adversaries: with a larger adversary, bounded
rollback is impossible \textcolor{black}{if the number of consistency
violations is bounded independently of $n$} (even in the synchronous setting).

\begin{theorem}[Impossibility result 2] \label{ir3} 
Consider the synchronous setting and suppose $\mathcal{P}$ is a protocol for SMR with liveness resilience $\rho_L$ and consistency resilience $\rho_C\geq \rho_L$.  If we are given only that the adversary is $\rho$-bounded for $\rho>1-\rho_L$, then $\mathcal{P}$ does not have \textcolor{black}{both} bounded rollback \textcolor{black}{and a bound on the number of consistency violations independent of $n$}. (The same result also holds in partial synchrony.) 
\end{theorem} 
\begin{proof}
{\color{black}
As in the proof of Theorem \ref{ir1}, let $\mathcal{F}$ be the function with respect to which $\mathcal{P}$ has liveness resilience $\rho_L$ and consistency resilience $\rho_C$. Suppose that $\mathcal{P}$ has bounded rollback against $\rho$-bounded adversaries. For any $r\in\mathbb{N}_{\geq1}$, we will construct a synchronous execution with at least $r$ consistency violations. We first construct two executions in which different sets of correct processes have finalized incompatible logs. We then enlarge the faulty sets, one process at a time, so that these logs give a nested sequence of message sets witnessing $r$ consistency violations.

\vspace{0.2cm}
Let $n$ be sufficiently large that there exists $m\in\mathbb{N}$ with $m<n$ and:
\[
1-\rho_L<m/n<\rho,\qquad m+r<\rho n,\qquad n-m>r+1.
\]
Let $\Pi=\{p_1,\dots,p_n\}$ and let $\Pi_1,\Pi_2,\Pi_3$ partition $\Pi$, with $|\Pi_1|=2m-n$ and $|\Pi_2|=|\Pi_3|=n-m$. Note that $n-m<\rho_Ln$. Also, $\rho_L\leq1/2$ by the synchronous resilience bound in Section \ref{setup}, so $\Pi_1$ is nonempty. Fix the empty genesis log, $\Delta^*=2$, and all other determined inputs, and let $h$ be the corresponding rollback bound. These inputs remain the same in all executions below. All processes start at time 0, and every message that is sent arrives at the next timeslot. Let $\mathtt{tr}_1$ and $\mathtt{tr}_2$ be distinct transactions. As in the proof of Theorem \ref{ir1}, we also use $\mathtt{tr}_1$ and $\mathtt{tr}_2$ for the corresponding one-transaction sequences. No transactions are supplied except as specified below.

\vspace{0.2cm}
\noindent \textbf{Execution} $\mathcal{E}_1$. Processes in $\Pi_1$ and $\Pi_2$ are correct. Processes in $\Pi_3$ are faulty and perform no actions. Processes in $\Pi_1$ receive $\mathtt{tr}_1$ at timeslot 1.

\vspace{0.1cm}
\noindent \textbf{Execution} $\mathcal{E}_2$. Processes in $\Pi_1$ and $\Pi_3$ are correct. Processes in $\Pi_2$ are faulty and perform no actions. Processes in $\Pi_1$ receive $\mathtt{tr}_2$ at timeslot 1.

\vspace{0.2cm}
Since $n-m<\rho_Ln$, all processes in $\Pi_2$ finalize $\mathtt{tr}_1$ by some timeslot $t_1$ in $\mathcal{E}_1$, and all processes in $\Pi_3$ finalize $\mathtt{tr}_2$ by some timeslot $t_2$ in $\mathcal{E}_2$. Since $\rho_C\geq\rho_L$, consistency holds in both executions. As only the corresponding transaction is supplied in each execution, all these processes finalize logs extending $\mathtt{tr}_1$ and $\mathtt{tr}_2$, respectively, at every later timeslot. Fix $t_3>\text{max}\{t_1,t_2\}+h$.

\vspace{0.2cm}
We next specify two executions $\mathcal{E}_3$ and $\mathcal{E}_4$ in which the correct processes see the same histories up to $t_3$ as in $\mathcal{E}_1$ and $\mathcal{E}_2$, respectively. After $t_3$, we arrange that each fixed process in $\Pi_2\cup\Pi_3$ has received the same set of messages in $\mathcal{E}_3$ as in $\mathcal{E}_4$.

\vspace{0.1cm}
\noindent \textbf{Execution} $\mathcal{E}_3$. Processes in $\Pi_2$ are correct. Processes in $\Pi_1$ and $\Pi_3$ are faulty, and processes in $\Pi_1$ receive both transactions at timeslot 1. At timeslots $\leq t_3$, the faulty processes proceed as follows:
\begin{enumerate}[beginpenalty=10000]
\item[(i)] Processes in $\Pi_1$ simulate their behaviour in $\mathcal{E}_1$, sending the corresponding messages to processes in $\Pi_1\cup\Pi_2$, but not those addressed to $\Pi_3$.
\item[(ii)] Processes in $\Pi_1\cup\Pi_3$ also simulate $\mathcal{E}_2$, sending the corresponding messages between members of $\Pi_1\cup\Pi_3$. Each process in $\Pi_1$ keeps the two simulated histories separate, using only the corresponding messages in each. Messages from this simulation addressed to $\Pi_2$ are not sent by processes in $\Pi_1$, while processes in $\Pi_3$ store them for sending at $t_3$. Processes in $\Pi_3$ record messages received from $\Pi_2$, but ignore them for this simulation.
\end{enumerate}
At timeslot $t_3$, processes in $\Pi_3$ send all the stored messages to their intended recipients, including those generated at $t_3$. Processes in $\Pi_1$ send no further messages after $t_3$.

\vspace{0.1cm}
\noindent \textbf{Execution} $\mathcal{E}_4$. This is defined symmetrically, interchanging $\Pi_2$ with $\Pi_3$ and $\mathcal{E}_1$ with $\mathcal{E}_2$. Processes in $\Pi_3$ are correct. Processes in $\Pi_1$ send messages from the $\mathcal{E}_2$ simulation only to $\Pi_1\cup\Pi_3$, and messages from the $\mathcal{E}_1$ simulation only to $\Pi_1\cup\Pi_2$. Processes in $\Pi_2$ ignore messages from $\Pi_3$ for their simulation and store messages addressed to $\Pi_3$, sending all the stored messages at $t_3$. Processes in $\Pi_1$ again receive both transactions at timeslot 1 and send no further messages after $t_3$.

\vspace{0.2cm}
To see why the received-message sets agree at $t_3+1$, fix $p_i\in\Pi_2$ and compare the messages it receives in $\mathcal{E}_3$ and $\mathcal{E}_4$:
\begin{enumerate}[beginpenalty=10000]
\item[(i)] Messages from $\Pi_1$ are precisely those prescribed by $\mathcal{E}_1$ in both executions, and arrive at the same timeslots. In particular, $\Pi_1$ sends no messages from its $\mathcal{E}_2$ simulation to $\Pi_2$ in either execution.
\item[(ii)] Messages from $\Pi_2$ also agree and arrive at the same timeslots. Through $t_3$, the correct processes in $\Pi_2$ in $\mathcal{E}_3$ reproduce their behaviour in $\mathcal{E}_1$, since no messages from $\Pi_3$ have yet arrived. The faulty processes in $\Pi_2$ in $\mathcal{E}_4$ simulate that same behaviour, ignoring messages from $\Pi_3$ for their simulation.
\item[(iii)] Messages from $\Pi_3$ are those prescribed by $\mathcal{E}_2$ through $t_3$ in both executions. In $\mathcal{E}_4$, the correct processes in $\Pi_3$ send them normally. In $\mathcal{E}_3$, the faulty processes in $\Pi_3$ store them and send them all at $t_3$. Thus, the same set of these messages has arrived by $t_3+1$, although their earlier receipt times may differ.
\end{enumerate}
No transactions are supplied directly to $p_i$, so these cases account for all its received messages. Hence $M_i(t_3+1)$ is the same in the two executions. The symmetric argument gives the same conclusion for each $p_i\in\Pi_3$.

We continue the two executions together, with each faulty process in $\Pi_2\cup\Pi_3$ simulating the same process in the execution in which it is correct. Its simulated history through $t_3$ already agrees with that correct process's history. At $t_3+1$, it supplies to its simulation the messages from the other group that it had previously ignored, matching their delivery to the correct process. Thus, the simulated local history, including receipt times, agrees with the correct process's history. Continuing timeslot by timeslot, corresponding messages are sent and received at the same timeslots, and processes in $\Pi_1$ send no further messages. So, for each fixed $p_i\in\Pi_2\cup\Pi_3$, $M_i(t)$ remains the same in the two executions for all $t\geq t_3+1$.
Every signed message used in a simulation has been received by the simulating process or is generated using a key it controls. Every message actually sent arrives at the next timeslot; only faulty processes postpone sending their stored messages.

\vspace{0.2cm}
Both executions have $m<\rho n$ faulty processes. The correct processes in $\Pi_2$ in $\mathcal{E}_3$ have finalized $\mathtt{tr}_1$ throughout an interval of length $h$, and those in $\Pi_3$ in $\mathcal{E}_4$ have done the same for $\mathtt{tr}_2$. Bounded rollback therefore requires their outputs eventually to extend the corresponding transactions permanently. Since each $p_i\in\Pi_2\cup\Pi_3$ has the same received-message set in both executions, the same conclusion holds for $\mathcal{F}(M_i(t))$ in the execution in which $p_i$ is faulty.

\vspace{0.2cm}
We now construct the nested message sets by induction. We begin with $\mathcal{E}_3$ and $\mathcal{E}_4$, but at each subsequent step replace them by new executions that retain the finite histories already obtained. The process made faulty at that step is treated as faulty from the start of each new execution, as explained below.

Choose distinct processes $p_{i_0},\dots,p_{i_r}$ in advance, with $p_{i_s}\in\Pi_2$ for even $s$ and $p_{i_s}\in\Pi_3$ for odd $s$. Choose $u_0\geq t_3+1$ sufficiently late that $p_{i_0}$'s received-message set $M_0$ in $\mathcal{E}_3$ satisfies $\mathcal{F}(M_0)\succeq\mathtt{tr}_1$. Its received-message set at $u_0$ in $\mathcal{E}_4$ is also $M_0$. Now suppose that, for some $s<r$, we have a pair of executions and a timeslot $u_s$ at which $p_{i_s}$ has received exactly $M_s$ in both, and is still correct in one. We maintain that each process in $\Pi_2\cup\Pi_3$ has the same received-message set in both executions at every timeslot $t\geq t_3+1$. We also maintain that $\mathcal{F}(M_s)$ extends $\mathtt{tr}_1$ for even $s$ and $\mathtt{tr}_2$ for odd $s$. We obtain $u_{s+1}$, $M_{s+1}$ and the next pair as follows:
\begin{enumerate}[beginpenalty=10000]
\item[(i)] Replace both executions by ones in which $p_{i_s}$ is faulty in both, without changing their histories through $u_s$. In each new execution the faulty set is fixed from the start; $p_{i_s}$ reproduces its previous behaviour in that execution up to $u_s$.
\item[(ii)] At $u_s+1$, $p_{i_s}$ sends all messages in $M_s$ to all processes, in both executions, and sends no further messages. The messages are available to it by $u_s$ and are delivered at $u_s+2$. Matching deliveries preserve equality of received-message sets, and the other simulations continue as before.
\item[(iii)] Process $p_{i_{s+1}}$ is still correct in one of the two executions. Because it finalized logs extending the opposite transaction throughout an interval of length $h$ before $t_3$, bounded rollback requires it eventually to finalize a log extending that transaction. Continue both executions. Choose $u_{s+1}\geq u_s+2$ late enough for this to hold. Let $M_{s+1}$ be the set of messages received by $p_{i_{s+1}}$ by $u_{s+1}$ in either execution; the two sets agree.
\end{enumerate}
At each step, bounded rollback is applied to complete continuations with the current faulty sets. We retain only the finite histories through $u_{s+1}$ before passing to the next pair. Throughout the construction, there are at most $m+r<\rho n$ faulty processes in either execution, and the chosen processes are available because $n-m>r+1$.

\vspace{0.2cm}
By construction, $M_s\subseteq M_{s+1}$, while $\mathcal{F}(M_s)$ and $\mathcal{F}(M_{s+1})$ extend different one-transaction sequences. So, the inclusion is strict and $\mathcal{F}(M_s)\not\preceq\mathcal{F}(M_{s+1})$. We stop before making $p_{i_r}$ faulty. In the final pair of executions, $p_{i_r}$ is still correct in one, and all the earlier sets $M_s$ have been sent to it. In that execution we have:
\[
M_0\subset M_1\subset\cdots\subset M_r\subseteq M_c.
\]
These finite sets witness $r$ consistency violations. Since $r$ was arbitrary, there can be no bound on the number of consistency violations independent of $n$, as required.
\par}
\end{proof}

{\color{black}
Our final impossibility result describes another sense in which our wrapper is optimal. \par}

\begin{theorem}[Impossibility result 3] \label{ir4}
\textcolor{black}{Recall the functions $g_1,g_2$ and the sequence $(x_r)_{r\geq0}$ from Section~\ref{rmsec}.}
Given $\rho_C$ and $\rho_L$ such that $\rho_C+2\rho_L=1$
\textcolor{black}{and $0<\rho_L\leq\rho_C=x_1\leq x_2\leq1-\rho_L$},
\textcolor{black}{suppose} $\textcolor{black}{g_2'\leq g_1'}\leq 1-\rho_L$ and that $\mathcal{P}$ is an SMR protocol for partial synchrony with  recoverable consistency resilience $\geq g_1'$ and recoverable liveness resilience $\geq g_2'$ when message delays are bounded by $\Delta^*$. Suppose that, for some $d$ and $\ell$, $\mathcal{P}$ has recovery time $d$ with liveness parameter $\ell$ \textcolor{black}{ with respect to $g_2'$} when the adversary is $1-\rho_L$-bounded. Then:
\begin{enumerate}
\item If $g_2'\geq g_2$, then $g_1'\leq g_1$, and;
\item If $g_2'>g_2$, then $g_1'<g_1$.
\end{enumerate}
\end{theorem}
\begin{proof}
{\color{black}
Suppose $g_2'\geq g_2$. Recall that $x_0=0$,
$x_{r+1}=x_r+\rho_C(1-x_r)$, and:
\[
g_1(r)=\min\{x_{r+1},1-\rho_L\},\qquad
g_2(r)=\min\{x_r+\rho_L(1-x_r),1-\rho_L\}.
\]
Under the hypotheses of the theorem, these definitions give:
\[
\begin{array}{ll}
g_1(0)=\rho_C, & g_2(0)=\rho_L,\\
g_1(1)=x_2, & g_2(1)=\rho_C+\rho_L(1-\rho_C).
\end{array}
\]
For $r\geq2$, both functions equal $1-\rho_L$. To see this, use
$1-\rho_C=2\rho_L$ to obtain $x_2=1-4\rho_L^2$, and hence:
\[
x_2+\rho_L(1-x_2)
=1-\rho_L+\rho_L(1-2\rho_L)^2\geq1-\rho_L.
\]
So $g_2(2)=1-\rho_L$. Both functions are nondecreasing, and
$g_1(r)\geq g_2(r)$ because $\rho_C\geq\rho_L$, so both remain at
this cap for $r\geq2$.
Since $g_1',g_2'\leq1-\rho_L$, the first conclusion is automatic for
$r\geq2$, and strict improvement in $g_2'$ can occur only at $r=0,1$.
We will prove the two bounds:
\begin{subequations}\label{ir4twobounds}
\begin{align}
g_1'(0)&\leq\max\{0,1-2g_2'(0)\},\label{ir4boundzero}\\
g_1'(1)&\leq\max\{\rho_C,1+2\rho_C-2g_2'(1)\}.\label{ir4boundone}
\end{align}
\end{subequations}
These suffice because $\rho_C+2\rho_L=1$ gives:
\[
\begin{split}
1-2g_2'(0)&=\rho_C-2\bigl(g_2'(0)-g_2(0)\bigr),\\
1+2\rho_C-2g_2'(1)&=x_2-2\bigl(g_2'(1)-g_2(1)\bigr).
\end{split}
\]
Also, $0<\rho_C<x_2$. Thus each bound in~\eqref{ir4twobounds} is at
most the corresponding value $g_1(0)$ or $g_1(1)$, and is strictly
smaller whenever that value of $g_2$ is improved.

\vspace{0.2cm} 
Bound~\eqref{ir4boundzero} is the usual consistency--liveness tradeoff,
adapted to our timing assumptions. We give a direct construction
because the proof of~\eqref{ir4boundone} requires more than the
numerical bound~\eqref{ir4boundzero}. We require a finite message set
$M'$, received by every correct process, such that for every finite
$N\supseteq M'$ there is a finite $A\subset N$ with
$\mathcal F(A)\not\preceq\mathcal F(N)$.
We first construct $M'$ and deduce~\eqref{ir4boundzero}.
Then, assuming~\eqref{ir4boundone} fails, we will use this property
to construct a second violation in the same nested chain, in an
execution with fewer than $ng_1'(1)$ faulty processes. This will
give the required contradiction.

\vspace{0.2cm}
\noindent\emph{Fixed conventions.}
In each construction below, choose $n$ and the group sizes before the
timing parameters. Fix $\Pi$, the empty genesis log, $\Delta=1$, and
the other determined inputs except $\Delta^*$. Having fixed $\ell$,
choose $\Delta^*>\ell+2$. All processes start at time~0. When a
construction is repeated for the same $n$, these inputs remain fixed.
A \emph{finite history} records the transaction inputs, process behaviour,
and message sends and receipts, with their times, through a finite timeslot.

\par\smallskip\noindent\emph{Obtaining the first violation.}
Take a partition of $\Pi$ into nonempty sets $\Pi_2,\Pi_3,\Pi_4$
with $|\Pi_3|,|\Pi_4|<ng_2'(0)$. 
We will construct two executions $\mathcal E_1,\mathcal E_2$ in
which $\Pi_2$ is correct and finalizes incompatible logs. In a third
execution $\mathcal E_3$, only $\Pi_2$ is faulty and its processes
disseminate the messages witnessing those logs. These messages will
give the required set $M'$.

Let $T\geq\max\{d,1\}$, where $d$ is the grace period in the
recovery-time guarantee, and put $D=T+\ell$. All three executions
have the same history through $T-1$: no transactions are supplied,
all processes follow the protocol, and messages arrive after one
timeslot. Choose distinct transactions $\mathtt{tr}_1,\mathtt{tr}_2$,
supplied only as specified below.
\begin{itemize}[beginpenalty=10000]
\item In $\mathcal E_1$, only $\Pi_4$ is faulty. Its processes
remain silent from $T$ onwards. Processes in $\Pi_2$ receive
$\mathtt{tr}_1$ at $T$, GST$=T$, and all messages have one-slot delay.
\item In $\mathcal E_2$, interchange $\Pi_3,\Pi_4$ and use
$\mathtt{tr}_2$ instead. Thus only $\Pi_3$ is faulty, and
$\Pi_2\cup\Pi_4$ is correct.
\end{itemize}
These faulty sets are fixed from time~0; they follow the protocol
through $T-1$ before becoming silent. Since
$|\Pi_3|,|\Pi_4|<ng_2'(0)\leq ng_1'(0)$, each complete execution
has zero violations. The recovery-time guarantee therefore applies
with the time of the zeroth violation defined to be~0. Its deadline
is $\max\{d,T,T\}+\ell=D$. Every process in $\Pi_2$ consequently
finalizes a log containing $\mathtt{tr}_i$ at $D$ in $\mathcal E_i$,
for $i=1,2$.

In $\mathcal E_3$, only $\Pi_2$ is faulty, and GST$=D+1$.
The required bound on $|\Pi_2|$ will be imposed when this construction
is applied below. At this stage, the recovery-time guarantee is used
only in $\mathcal E_1$ and $\mathcal E_2$.
Processes in $\Pi_2$ receive both transactions at $T$. Through $D$ they
maintain separate simulations of their behaviour in $\mathcal E_1$
and $\mathcal E_2$:
\begin{itemize}[beginpenalty=10000]
\item Messages prescribed by the first simulation are sent to
$\Pi_2\cup\Pi_3$, but not to $\Pi_4$; those prescribed by the
second are sent to $\Pi_2\cup\Pi_4$, but not to $\Pi_3$.
Messages within $\Pi_2$ are supplied only to their own simulation.
\item Messages to $\Pi_2$ still in transit at $T$ arrive at $T$ and are
supplied to both simulations, as in their common earlier history.
Messages sent by $\Pi_3$ from $T$ onwards are supplied only to the
first simulation, and those sent by $\Pi_4$ only to the second:
the other group is silent in the corresponding execution.
\end{itemize}
Messages sent between $\Pi_3$ and $\Pi_4$ during $[T,D]$ are delayed
until $D+1$; all other messages have one-slot delay. Thus, through
$D$, processes in $\Pi_3$ have their histories from $\mathcal E_1$:
messages from $\Pi_2$ come only from the first simulation, and new
messages from $\Pi_4$ have not arrived. Symmetrically, processes in
$\Pi_4$ have their histories from $\mathcal E_2$. The longest delay
is $D+1-T=\ell+1<\Delta^*$.

We now specify the rest of $\mathcal E_3$ by disseminating the
incompatible certificates. Choose $p_i\in\Pi_2$, and let $M_1,M_2$ be its sets
of received messages at $D$ in the two simulations (not its full
actual set in $\mathcal E_3$). The log $\mathcal F(M_1)$ contains
$\mathtt{tr}_1$ but not $\mathtt{tr}_2$, which was never issued in
$\mathcal E_1$, while  the reverse holds for $\mathcal F(M_2)$. The logs
are therefore incompatible. All messages in $M_1\cup M_2$ have been
received by processes in $\Pi_2$ or generated using their keys.
At $D+1$ they send $M'=M_1\cup M_2$ to every correct process,
which receives it at $D+2$. Thereafter faulty processes remain
silent, and correct processes continue with one-slot message delays.

\vspace{0.2cm} 
\noindent\emph{Analysis of $\mathcal E_3$.}
For every finite $N\supseteq M'$ there is a $j\in\{1,2\}$ such that
$\mathcal F(M_j)\not\preceq\mathcal F(N)$: the two incompatible
logs cannot both be prefixes of $\mathcal F(N)$. So
$M_j\subset N$ witnesses one violation, and the inclusion is strict
because equality would make the prefix relation hold. We have therefore
obtained a finite history through $D+2$ in which every correct process
has received $M'$, and every finite superset of $M'$ ends a
one-violation chain. The faulty set in this history is $\Pi_2$.

\par\smallskip\noindent\emph{Consequence for $r=0$.}
Suppose~\eqref{ir4boundzero} fails, and choose
\[
\max\{0,1-2g_2'(0)\}<\rho<g_1'(0).
\]
For sufficiently large $n$, choose $|\Pi_2|/n$ strictly between the
displayed lower bound and $\rho$, then split the remaining processes
as evenly as possible. The strict inequalities allow for integer
rounding and give nonempty $\Pi_3,\Pi_4$ with
$|\Pi_3|,|\Pi_4|<ng_2'(0)$. Apply the construction with these sizes.
Taking $N=M'$ gives a violation in $\mathcal E_3$ with
$|\Pi_2|<\rho n<ng_1'(0)$ faulty processes, a contradiction.

We now use the stronger property of $M'$ to obtain a second violation,
not merely another incompatible pair of logs.

\par\smallskip\noindent\emph{The bound at $r=1$.}
Suppose~\eqref{ir4boundone} fails, and choose:
\[
\max\{\rho_C,1+2\rho_C-2g_2'(1)\}<\rho<g_1'(1).
\]
We will choose integers $f_1<f_2$ for the numbers of faulty processes
in the first construction and the final execution, respectively.
We first construct a finite history containing a violation with
$f_1$ faulty processes. We then reproduce this history in an execution
with $f_2$ processes faulty from the start, and construct a second
violation. We will have $f_2<\rho n<ng_1'(1)$, giving the required
contradiction. Choose $n$ sufficiently large and integers
$0<f_1<f_2<\rho n$, with $n-f_2\geq2$, such that:
\[
\left\lceil\frac{n-f_1}{2}\right\rceil<n\rho_L,
\qquad
f_1+\left\lceil\frac{n-f_2}{2}\right\rceil<ng_2'(1).
\]
To obtain these inequalities, choose $f_1/n$ just above
$\rho_C=1-2\rho_L$ and $f_2/n$ just below $\rho$. This is possible
with strict margins for rounding because $\rho>\rho_C$ and
$\rho_C+(1-\rho)/2<g_2'(1)$. Fix these choices before the timing
parameters, which remain the same throughout.

\vspace{0.2cm}
\noindent\emph{Obtaining the first violation for $r=1$.}
Apply the construction of $\mathcal E_3$ under ``Obtaining the first
violation'' with $|\Pi_2|=f_1$ and the other
two groups as evenly sized as possible. The first size inequality
gives $|\Pi_3|,|\Pi_4|<n\rho_L\leq ng_2'(0)$, as required.
Write $M$ for the resulting set $M'$, and $H$ for the history of
$\mathcal E_3$ through its timeslot $D+2$. Every process correct
in $H$ has received $M$, and every finite superset of $M$ ends a
one-violation chain.

\vspace{0.2cm}
\noindent\emph{Constructing the second violation.}
We now construct new executions extending $H$. Rename its faulty
set $\Pi_1$, choose $f_2-f_1$ of the
remaining processes as the new $\Pi_2$, and split the rest as evenly
as possible into nonempty sets $\Pi_3,\Pi_4$. The second size
inequality gives:
\[
|\Pi_1\cup\Pi_3|,\ |\Pi_1\cup\Pi_4|<ng_2'(1)\leq ng_1'(1).
\]
These will be the faulty sets in new executions $\mathcal E_2$ and
$\mathcal E_1$, respectively. The execution names, group names and
times below refer to this new construction.

Let $t_0$ be the last timeslot of $H$. Extend $H$ with no new
transactions, processes in $\Pi_1$ sending no messages, all other
processes following the protocol, and new messages delivered after
one timeslot. Choose $T>t_0$ and
$T\geq t_0+d$, after all messages sent by $t_0$ have arrived.
Retain this common history through $T-1$, put $D=T+\ell$, and
choose two previously unissued transactions, again denoted
$\mathtt{tr}_1,\mathtt{tr}_2$.
\begin{itemize}[beginpenalty=10000]
\item In the new $\mathcal E_1$, $\Pi_1\cup\Pi_4$ is faulty.
Processes in $\Pi_4$ follow their behaviour in the common history,
then remain silent from $T$ onwards. Processes in $\Pi_2$ receive
$\mathtt{tr}_1$ at $T$, GST$=T$, and new messages have one-slot delay.
\item The new $\mathcal E_2$ is defined symmetrically, with
$\Pi_1\cup\Pi_3$ faulty and $\mathtt{tr}_2$ supplied instead.
\end{itemize}
In both executions, processes in $\Pi_1$ send no messages after $H$,
and no other new transactions are supplied. Each faulty set is fixed from the
start; newly designated faulty processes reproduce their earlier
correct behaviour throughout the common history. GST is not part
of that history and may be chosen later, since it is unknown to
the processes.

\vspace{0.2cm} 
\noindent\emph{Analysis of $\mathcal E_1$ and $\mathcal E_2$.}
Each of these complete executions has at most one violation, by
its fault bound and $g_2'(1)\leq g_1'(1)$. It also has at least one:
its correct processes received $M$ during $H$, and those receipt
events, together with the one-violation witness in $M$, are preserved.
Thus each has precisely one violation, already witnessed by $t_0$.
If $t_i$ is the first time it is witnessed in $\mathcal E_i$, then
$t_i\leq t_0$ and $t_i+d\leq T$. The recovery-time guarantee applies
because the faulty fraction is less than $g_2'(1)\leq1-\rho_L$.
Its deadline is again $\max\{t_i+d,T,T\}+\ell=D$, so processes in
$\Pi_2$ finalize the respective fresh transactions by $D$.

\vspace{0.2cm} 
\noindent\emph{The new execution $\mathcal E_3$.}
Let $\mathcal E_3$ be specified as in the paragraph ``Obtaining the
first violation'', using the new groups and executions, with the
following details:
\begin{itemize}[beginpenalty=10000]
\item The faulty set is $\Pi_1\cup\Pi_2$, fixed from the start.
The history through $T-1$ is retained, and GST$=D+1$.
Processes in $\Pi_1$ send no messages after $H$.
\item Processes in $\Pi_2$ receive both fresh transactions at $T$
and carry out the two simulations of their behaviour in the new
$\mathcal E_1$ and $\mathcal E_2$, with the same sending and
message-delivery rules as in that paragraph. Messages from either
simulation addressed to $\Pi_1$ are also sent. All messages sent to
$\Pi_1$ arrive after one timeslot. Any messages to $\Pi_2$ in transit at $T$
were sent at $T-1$ and are supplied to both simulations on receipt.
\item Choose $p_i\in\Pi_2$, and let $M_1,M_2$ be its two simulated
received-message sets at $D$. At $D+1$, processes in $\Pi_2$ send
$N=M_1\cup M_2$ to every correct process, which receives it at
$D+2$. Thereafter faulty processes send no messages, and correct
processes continue with one-slot message delays.
\end{itemize}

\vspace{0.2cm}
\noindent\emph{Analysis of $\mathcal E_3$.}
As in the first construction of $\mathcal E_3$, the correct processes
in $\Pi_3,\Pi_4$ have their respective histories from
$\mathcal E_1,\mathcal E_2$ through $D$. The two simulated
received-message sets $M_1,M_2$ therefore give incompatible
finalized logs. All messages in $N$ have been received by processes
in $\Pi_2$ or generated using their keys, so these processes can
send $N$ as specified. The entire execution respects $\Delta^*$.

The selected process in $\Pi_2$ received $M$ during $H$, so both
$M_1,M_2$ contain $M$. Choose $j$ such that
$\mathcal F(M_j)\not\preceq\mathcal F(N)$. By the property of $M$,
there is a finite set $A\subset M_j$ with
$\mathcal F(A)\not\preceq\mathcal F(M_j)$. Recall that $M_c(D+2)$
is the union of all messages received by correct processes by $D+2$
in the new $\mathcal E_3$. Since every such process has received $N$,
\[
A\subset M_j\subset N\subseteq M_c(D+2)
\]
is a two-violation chain in this execution. It has $f_2<\rho n<ng_1'(1)$ faulty processes, a
contradiction. Both bounds in~\eqref{ir4twobounds}, and hence both
conclusions, follow.
\par}
\end{proof}

\section{The proof of Theorem \ref{mt}} \label{pmt} 

\textcolor{black}{First consider the usual consistency and liveness guarantees in partial synchrony. If the adversary is $\rho_C$-bounded, no consistency violation occurs in $\mathcal{E}_1$, so recovery is never triggered. Since $\rho_L\leq\rho_C$, this also applies whenever the adversary is $\rho_L$-bounded, and the wrapper inherits the liveness guarantee of $\mathcal{P}$ without requiring the bound $\Delta^*$ before GST.}

We assume throughout this section that the adversary is
$(1-\rho_L)$-bounded %
and that message delays are bounded by $\Delta^*$. 

{\color{black}
\vspace{0.2cm}
\noindent \textbf{Structure of the verification.}
The proof follows the tasks explained in Section~\ref{intuition}.
First, Lemma~\ref{ml} establishes synchronized entry into recovery,
agreement on a unique proposal, termination with a correct leader,
and a common restart (Section~\ref{agreeintuition}).
Next we show that the local log maintained by a correct process agrees
with the finalization function in Algorithm~2.
Lemma~\ref{cons} then establishes that strongly finalized logs are
preserved when recovery takes place, giving the rollback bound
(Sections~\ref{ebr} and~\ref{rpi}).
Finally, Lemma~\ref{vbound} relates consistency violations to
recovery procedures, and the last two lemmas establish the resilience
and recovery-time bounds.
}

\vspace{0.2cm} 
\noindent \textbf{Some further terminology}. We make the following definitions: 
\begin{itemize} 
\item Process $p_i$ begins the $r^{\text{th}}$ execution of the recovery procedure at the first timeslot at which  $\mathtt{r}=r$ and $\mathtt{rec}=1$ (where those values are as locally defined for $p_i$). 
\item The $r^{\text{th}}$ execution of the recovery procedure  begins at the first timeslot at which some correct process begins the  $r^{\text{th}}$ execution of the recovery procedure. 
\item Execution $\mathcal{E}_r$ begins at the first timeslot at which some correct process begins execution $\mathcal{E}_r$. If $r>1$, then the $(r-1)^{\text{th}}$ execution of the recovery procedure also ends at this timeslot. 
\item If a QC/finish-QC is regarded as valid by all correct processes, we refer to it as a valid QC/finish-QC. 
\end{itemize}

\begin{lemma} \label{ml} 
If the $r^{\text{th}}$ execution of the recovery procedure begins at $t_0$, then: 
\begin{enumerate} 
\item[(i)] All correct processes begin the $r^{\text{th}}$ execution of the recovery procedure by time $t_0+\Delta^*$. 
\item[(ii)] There exists a unique $r$-proposal, $P$ say,  that receives a finish-QC that is regarded as valid by some correct process. 
\item[(iii)] If $v_0$ is least such that $\mathtt{lead}(r,v_0)$ is correct, all correct processes receive a valid finish-QC for $P$ by time $t_0+2\Delta^*+8v_0\Delta^*$. 
\item [(iv)] All correct processes begin execution $\mathcal{E}_{r+1}$ within time $\Delta^*$ of each other and with the same local values for $\Pi_{r+1}$ and $\mathtt{log}_{G_{r+1}}$.  
\end{enumerate} 
\end{lemma} 
\begin{proof} 
\par\smallskip\noindent\textcolor{black}{\textit{Coordinated entry into recovery.}}
The proof is by induction on $r$. So, suppose the lemma  holds for all $s<r$. 
The  $r^{\text{th}}$ execution of the recovery procedure cannot begin before execution $\mathcal{E}_r$. 
So, since the lemma holds for  $r-1$ (if $r>1$), it follows that all correct processes begin $\mathcal{E}_r$  by $t_0+\Delta^*$. All messages received by the first correct process to begin the $r^{\text{th}}$ execution of the recovery procedure will also be received by all correct processes by $t_0+\Delta^*$, meaning that all correct processes will begin the $r^{\text{th}}$ execution of the recovery procedure by that time (and within $\Delta^*$ of each other). This establishes statement (i) of the lemma.

\par\smallskip\noindent\textcolor{black}{\textit{Uniqueness of the agreed proposal.}}
 By the induction hypothesis, all correct processes have
the same local values for $\Pi_{j}$ and $\mathtt{log}_{G_{j}}$, so
that any QC or finish-QC that is regarded as valid by any correct
process will be regarded as valid by all correct processes \textcolor{black}{ once the required proofs of guilt have been received}. The
remainder of the proof of Lemma~\ref{ml} proceeds much as in the
\textcolor{black}{informal discussion in Section~\ref{agreeintuition}}. Below, we
fill in the details.

\vspace{0.1cm} 
Suppose first that there exists some least $v$ such that some correct process  $p_i$ sends  a finish vote for some $r$-proposal,  $P$ say, during view $(r,v)$. Then, at some timeslot $t$ while in view $(r,v)$,  $p_i$ must set $Q_i^+:=Q$, for some $Q$ such that: (i) $Q$  is a valid QC for an $(r,v)$-proposal $R$, and; (ii)  $\mathtt{P}(Q)=P$.
 We will show that: 
\begin{enumerate} 
\item All correct processes set their locks to some valid QC for $R$ while in view $(r,v)$.
\item No $(r,v)$-proposal other than $R$ can receive a QC that is regarded as valid by any correct process. 
\end{enumerate}

\vspace{0.2cm} 
To see that (1) holds, note that all correct processes will be in view $(r,v)$ at $t+\Delta^*$, since $p_i$ is still in view $(r,v)$ at time $t+2\Delta^*$ and all correct processes begin the recovery procedure (and so each view of the recovery procedure) within time $\Delta^*$ of each other. All correct processes will also have received $Q$ \textcolor{black}{ and the proofs of guilt justifying $F$} by  $t+\Delta^*$. 
 They will therefore set their lock to be some QC for $R$, unless they have already received a valid QC for some $(r,v)$-proposal $R'\neq R$. The latter case is not possible, since then $p_i$ would detect equivocation in view $(r,v)$ by $t+2\Delta^*$, and so would not send the finish vote for $P$. 

\vspace{0.2cm} 
To see that (2) holds, note that (as reasoned above) all correct processes will be in view $(r,v)$ at $t+\Delta^*$ and will have received $R$ by this time. Item (ix) in the validity conditions for $(r,v)$-proposals prevents correct processes  voting for $(r,v)$-proposals  $R'\neq R$ at  timeslots after $t+\Delta^*$, and correct processes cannot vote for such proposals at any timeslot $\leq t+\Delta^*$ because $p_i$ would detect equivocation in view $(r,v)$ before sending the finish vote for $P$ in this case. Since any QC for an $(r,v)$-proposal that is regarded as valid by any correct process must include at least one vote by a correct process, the claim follows. 

\vspace{0.2cm} 
From (2) it follows that no correct process sends a finish vote for any $r$-proposal other than $P$ while in view $v$. Next, we  argue by induction on $v'>v$ that no correct process votes for any proposal $R'=(P',v',Q')_{p_k}$ such that $P'\neq P$. So, suppose the claim holds for all $v''\in (v,v')$, meaning that no $R'=(P',v'',Q')_{p_k}$ with $P'\neq P$ and $v''\in (v,v')$ can receive a valid QC. From (1), it follows that every correct process $p_j$ has their local lock $Q_j^+$ set so that $\mathtt{view}(Q_j^+)\geq (r,v)$ upon entering view $v'$. They therefore cannot vote for any proposal $R'=(P',v',Q')_{p_k}$ such that $P'\neq P$, since for $R$  to be valid it must hold that $\mathtt{P}(Q')=P'$, so that $\mathtt{view}(Q')<(r,v)$, while $\mathtt{view}(Q')\geq \mathtt{view}(Q_j^+)\geq (r,v)$.

\par\smallskip\noindent\textcolor{black}{\textit{Termination with a correct leader.}}
 So far, we have established that at most one
$r$-proposal receives a valid finish-QC. To establish (ii) and (iii)
in the statement of the lemma, let $t_0$ and $v_0$ be as specified in the
statement of the lemma. Note that $t_0+2\Delta^*+8v_0\Delta^*$ is the
first timeslot at which any correct process can enter view
$(r,v_0+1)$. Suppose, towards a contradiction, that it is not the case
that all correct processes receive a valid finish-QC for some
$r$-proposal by $t_0+2\Delta^*+8v_0\Delta^*$. This means that no
correct process receives a valid finish-QC for any $r$-proposal by
$t_0+\Delta^*+8v_0\Delta^*$, so all correct processes are still
executing the recovery procedure and are in view $(r,v)$ at this time.
 
Let $i$ be such that $p_i=\mathtt{lead}(r,v_0)$.  Since $p_i$ waits $2\Delta^*$ after entering the view before disseminating an $(r,v_0)$-proposal $R=(P,v_0,Q)_{p_i}$, and since all correct processes enter the view at most $\Delta^*$ after $p_i$, it will have seen all locks held by correct processes by this time, and will have received $r$-genesis messages from all processes in any set $P_{j}(r)$ for correct $p_j$. \textcolor{black}{A new proposal therefore meets condition (v) for every correct process; a proposal carried forward with a valid QC for the same $r$-proposal need not meet that local genesis-message condition again.} Process $p_i$ will disseminate the $(r,v_0)$-proposal $R$ by  $t_0+5\Delta^* +8(v_0-1)\Delta^*$, since $p_i$  begins the $r^{\text{th}}$ execution of the recovery procedure by time $t_0+\Delta^*$. All correct processes  will have received the  proposal by timeslot  $t_0+6\Delta^* +8(v_0-1)\Delta^*$ and will regard the proposal as valid by this time (having received all $(\mathcal{P},\Pi_r,\mathtt{log}_{G_r})$-proofs of guilt that $p_i$ received by $t_0+5\Delta^* +8(v_0-1)\Delta^*$). All correct processes will  therefore vote for the proposal by this time and will receive a valid QC for the proposal by $t_0+7\Delta^* +8(v_0-1)\Delta^*$, and so will set their locks by this time.  They will still be in view $(r,v)$ by time $t_0+9\Delta^*+8(v_0-1)\Delta^*= t_0+\Delta^*+8v_0\Delta^*$  and will  send $r$-finish votes for $P$ by this time. All correct processes will therefore receive a valid finish-QC for $P$ by time $t_0+2\Delta^* + 8v_0\Delta^*$, giving the required contradiction.

\par\smallskip\noindent\textcolor{black}{\textit{Common restart.}}
To see that statement (iv) of the lemma holds, note that the first correct process to begin execution $\mathcal{E}_{r+1}$ does so upon receiving a valid finish-QC for some $r$-proposal $P$ at some timeslot $t$.  This cannot happen until at least $2\Delta^*$ has passed since the $r^{\text{th}}$ execution of the recovery procedure began, meaning that all correct processes will have begun the $r^{\text{th}}$ execution of the recovery procedure by timeslot $t$. They will all receive a valid finish-QC for $P$ by $t+\Delta^*$, and so will begin execution $\mathcal{E}_{r+1}$ by that time. 
\end{proof} 

\noindent \textbf{Further notation}. Given statement (iv) of Lemma \ref{ml}, each value $\Pi_r$ or $\mathtt{log}_{G_r}$ is either undefined at all timeslots for all correct processes, or else is eventually defined and takes the same value for each correct process. We may therefore write $\Pi_r$ and $\mathtt{log}_{G_r}$ to denote these globally agreed values.

\begin{lemma} 
If $p_i$ is correct, then, at the end of every timeslot,  $\mathtt{log}_i=\mathcal{F}(\mathtt{M}_i)$. 
\end{lemma} 
\begin{proof} 
Let $\mathtt{rec}$, $\mathtt{r}$, $\mathtt{M}_i$, $\mathtt{M}_{i,r}$ and $\mathtt{log}_i$ be as locally defined for $p_i$. Consider first the case that $\mathtt{rec}=0$ at the end of timeslot $t$. \textcolor{black}{If $p_i$ starts a new execution of $\mathcal{P}$ at this timeslot, the equality holds by the final line of Algorithm~1. Otherwise the following argument applies.} In this case, $\mathtt{M}_i$ has a consistency violation with respect to $\mathcal{F}(\Pi_r,\mathtt{log}_{G_r})$ for each $r<\mathtt{r}$ but does not have a consistency violation with respect to $\mathcal{F}(\Pi_\mathtt{r},\mathtt{log}_{G_\mathtt{r}})$. Also, $\mathtt{M}_{i}$ contains a valid finish-QC for some $r$-proposal for each $r<\mathtt{r}$, which must be unique by Lemma \ref{ml}. At the end of timeslot $t$, $\mathtt{log}_i$ is the longest string $\sigma$ such that $\mathtt{M}_i$ (and $\mathtt{M}_{i,\mathtt{r}}$) is an  $\mathcal{F}(\Pi_\mathtt{r},\mathtt{log}_{G_\mathtt{r}})$-certificate for $\sigma$. The iteration defining $\mathcal{F}$ in Algorithm 2 will not return a value until it has defined all values $\Pi_r$ and $\mathtt{log}_{G_r}$ for $r\leq \mathtt{r}$. Upon discovering that $\mathtt{M}_i$ does not have a consistency violation with respect to $\mathcal{F}(\Pi_\mathtt{r},\mathtt{log}_{G_\mathtt{r}})$, it will return the same value $\sigma$, as the longest string for which  $\mathtt{M}_i$ is an  $\mathcal{F}(\Pi_\mathtt{r},\mathtt{log}_{G_\mathtt{r}})$-certificate. 

Next, consider the case that  $\mathtt{rec}=1$ at the end of timeslot $t$.  In this case, $\mathtt{M}_i$ has a consistency violation with respect to $\mathcal{F}(\Pi_r,\mathtt{log}_{G_r})$ for each $r\leq \mathtt{r}$, and also  contains a valid finish-QC for some $r$-proposal for each $r<\mathtt{r}$, which must be unique by Lemma \ref{ml}. However, $\mathtt{M}_i$ does not contain  a valid finish-QC for any  $\mathtt{r}$-proposal.  At the end of timeslot $t$, $\mathtt{log}_i=\mathtt{log}_{G_\mathtt{r}}$. The iteration defining $\mathcal{F}$ will not  return a value until it has defined all values $\Pi_r$ and $\mathtt{log}_{G_r}$ for $r\leq \mathtt{r}$, and will then also return $\mathtt{log}_{G_\mathtt{r}}$.
\end{proof} 

\begin{lemma} \label{cons} 
 The wrapper has rollback bounded by $2\Delta^*$. Also, $\mathtt{log}_{G_{r+1}}\succeq \mathtt{log}_{G_r}$ whenever $\mathtt{log}_{G_{r+1}}\downarrow$.   
\end{lemma} 
\begin{proof} 
\par\smallskip\noindent\textcolor{black}{\textit{Propagation of finalization.}}
We say $p_i$ finalizes $\sigma$ if it sets $\mathtt{log}_i$ to extend $\sigma$ and that $p_i$ strongly finalizes $\sigma$ if it sets $\mathtt{log}_i^*$ to extend $\sigma$. Suppose  $p_i$ finalizes $\sigma$  while running $\mathcal{E}_r$ at $t$ because there exists $M\subseteq \mathtt{M}_{i,r}$ which is an $\mathcal{F}(\Pi_r,\mathtt{log}_{G_r})$-certificate for $\sigma$. By (iv) of Lemma \ref{ml},  every correct process $p_j$ will begin $\mathcal{E}_r$ by $t+\Delta^*$, and will receive the messages in $M$ by that time. This means $p_j$ will finalize $\sigma$, never to subsequently finalize any sequence  incompatible with $\sigma$ while running $\mathcal{E}_r$, unless  $\mathtt{M}_{j,r}$ has a consistency violation w.r.t.\ $\mathcal{F}(\Pi_r,\mathtt{log}_{G_r})$ by  $t+\Delta^*$.   In the latter case, $p_i$ will begin the recovery procedure by timeslot $t+2\Delta^*$ and will not strongly finalize $\sigma$ at that time. We conclude that, if $p_i$ strongly finalizes $\sigma$ while running $\mathcal{E}_r$, then all correct processes finalize $\sigma$ while running $\mathcal{E}_r$. 

\par\smallskip\noindent\textcolor{black}{\textit{Preservation during recovery.}}
If $p_i$ strongly finalizes $\sigma$ while running $\mathcal{E}_r$ and if the $r^{\text{th}}$ execution of the recovery procedure does not begin at any timeslot, it follows that, for all correct $p_j$,  $\sigma \preceq \mathtt{log}_j$ thereafter. So, suppose that the $r^{\text{th}}$ execution of the recovery procedure begins at some timeslot $t_0$. 
Note that, if $p_j$ is correct, then it waits $2\Delta^*$ after beginning the $r^{\text{th}}$ execution of the recovery procedure before defining $P_j(r)$. By Lemma \ref{ml}, all correct processes begin the $r^{\text{th}}$ execution of the recovery procedure within time $\Delta^*$ of each other. Since correct processes send $r$-genesis messages immediately upon beginning the recovery procedure, it follows that $P_j(r)$ includes all correct processes. 

By Lemma \ref{ml}, there exists a unique $r$-proposal, $P=(F,\sigma',M',r)$ say, that receives a valid finish-QC. \textcolor{black}{Let $v$ be the least view in which $P$ receives a valid QC. Since a proposal carrying a QC for $P$ must carry one from an earlier view, an $(r,v)$-proposal $R=(P,v,\bot)$ (signed by $\mathtt{lead}(r,v)$) receives a valid QC.} This QC must include at least one vote by a correct process, $p_j$ say. It follows that $M'$ must contain $r$-genesis messages from every member of $P_j(r)\textcolor{black}{-F}$, and so from every correct process\textcolor{black}{, since $F$ contains only faulty processes}. As we noted previously, if $p_i$ strongly finalizes $\sigma$ while running $\mathcal{E}_r$, then every correct process must finalize $\sigma$ before beginning the $r^{\text{th}}$ execution of the recovery procedure. 
So, for each $r$-genesis message $(\text{gen},\sigma'',r)$ sent by a correct process, $\sigma''$ must extend $\sigma$, and also extends $\mathtt{log}_{G_r}$. It therefore holds that $\sigma $ (and $\mathtt{log}_{G_r}$) is extended by more than $\frac{1}{2}|\Pi_r-F|$ elements of $M'$. Recall that $\sigma'$ is as specified by $P$. No correct process will vote for $R$ unless  $\sigma'$ is the longest sequence  extended by more than $\frac{1}{2}|\Pi_r-F|$ elements of \textcolor{black}{$M'$}, meaning that $\sigma'$ must extend $\sigma$ and $\mathtt{log}_{G_r}$. So far, we conclude that $\mathtt{log}_{G_{r+1}}$ extends $\sigma$ and $\mathtt{log}_{G_r}$. Since it follows by the same argument that $\mathtt{log}_{G_{s}}$ extends $\sigma$ for all $s>r$, the claim of the lemma holds.     
\end{proof}

\begin{lemma} \label{vbound} 
If an execution of the wrapper has $r$ consistency violations, then the $r^{\text{th}}$ execution of the recovery procedure must begin at some timeslot (and so, by Lemma \ref{ml} must  also end at some timeslot). \textcolor{black}{More precisely, if $r\geq1$ and $M_c(t)$ has $r$ consistency violations, and $v_0$ is the first view with a correct leader in recovery $r$, then all correct processes begin $\mathcal{E}_{r+1}$ by $t+(3+8v_0)\Delta^*+1$.} 
\end{lemma} 
\begin{proof} 
For any set of messages $M$, let $r^*(M)$ be the greatest value taken by the variable $\mathtt{r}$ when Algorithm 2 is run with input $M$. From (ii) of Lemma \ref{ml}, it follows that if $M_0\subseteq M_1$ then $r^*(M_0)\leq r^*(M_1)$. 
If an execution of the wrapper has $r$ consistency violations, then there exist $M_0 \subset M_1 \subset \dots \subset M_r\subseteq M_c$ such that, for all $r'\in [0,r)$, $\mathcal{F}(M_{r'}) \not \preceq  \mathcal{F}(M_{r'+1})$. If there exists $r'<r-1$ and $s$ such that $s=\textcolor{black}{r^*}(M_{r'})=\textcolor{black}{r^*}(M_{r'+1})$, then $\mathcal{F}(M_{r'+1} )=\mathtt{log}_{G_s}$, and, by Lemma \ref{cons}, $\mathcal{F}(M) \succeq \mathcal{F}(M_{r'+1} )$ for all $M\supseteq M_{r'+1}$. This contradicts the fact that $\mathcal{F}(M_{r'+1} ) \not \preceq \mathcal{F}(M_{r'+2})$. 
We conclude that,  for all $r'\in [0,r-1)$, $r^*(M_{r'})<r^*(M_{r'+1})$. If $r^*(M_{r-1})=r^*(M_{r})$ then $r^*(M_r)$ is at least $r$ and $M_r$ has a consistency violation with respect to $\mathcal{F}(\Pi_r,\mathtt{log}_{G_r})$, meaning that the $r^{\text{th}}$ execution of the recovery procedure must begin at some timeslot, as claimed. If not, then  $r^*(M_r)>r$, which also means that the $r^{\text{th}}$ execution of the recovery procedure must begin at some timeslot.
\textcolor{black}{For the time bound, take $M_r\subseteq M_c(t)$. If $r^*(M_r)>r$, then $M_r$ contains a valid $r$-finish-QC, with a finish vote sent by a correct process by time $t$.  So,  recovery $r$ has begun by $t$. Otherwise $r^*(M_r)=r$, and $M_r$ contains both the finish-QCs for earlier recoveries and a consistency violation in $\mathcal{E}_r$. When $r=1$, all correct processes have started $\mathcal{E}_1$ by $t+\Delta^*$. When $r>1$, a correct voter in an $(r-1)$-finish-QC in $M_r$ has entered recovery $r-1$ by time $t$. By $t+\Delta^*$ it has received this finish-QC and the other messages in $M_r$, and has started $\mathcal{E}_r$ if it has not already progressed further. It therefore enters recovery $r$ by $t+\Delta^*+1$, allowing a timeslot to test for the violation. So,  in all cases recovery $r$ begins by $t+\Delta^*+1$. Lemma~\ref{ml}(iii),(iv) gives the stated bound for starting $\mathcal{E}_{r+1}$.}
\end{proof} 

\begin{lemma}
The wrapper has recoverable consistency resilience $\geq g_1$ and also recoverable liveness resilience $\geq g_2$.
\end{lemma} 
\begin{proof} 
Recall that, in Section \ref{rmsec}, we set $ x_0=0$ and $x_{r+1}=x_r+\rho_C(1-x_r)$, and then defined:
\[ g_1(r)=\text{min} \{ x_{r+1}, 1-\rho_L \}, \hspace{0.3cm}  g_2(r)=\text{min} \{ x_r+\rho_L(1-x_r),1-\rho_L \}. \]
Note that $x_r$ lower bounds the fraction of the processes removed to form $\Pi_{r+1}$, i.e. $|\Pi-\Pi_{r+1}|\geq x_r n$.  \textcolor{black}{By Lemma}~\ref{vbound}, if there exist $r$ consistency violations, then the $r^{\text{th}}$ execution of the recovery procedure must end at some timeslot. If the adversary is $g_2(\textcolor{black}{r})$-bounded, and since $\mathcal{P}$ has liveness resilience $\rho_L$, it follows that liveness must hold. If the adversary is $g_1(r)$-bounded, then  since $\mathcal{P}$ has consistency  resilience $\rho_C$, the $(r+1)^{\text{th}}$ execution of the recovery procedure cannot begin at any timeslot. From Lemma \ref{vbound}, it follows that there are at most $r$ consistency violations. 
\end{proof} 

\begin{lemma}
 The wrapper has recovery time $O(f_a\Delta^*)$ with liveness parameter $\ell$, where $f_a$ is the actual (unknown) number of faulty processes.
It also has probabilistic recovery time $O(\Delta^*\text{log}\frac{1}{\varepsilon})$ with liveness parameter $\ell$. \textcolor{black}{Both bounds are with respect to $g_2$.}
\end{lemma}
\begin{proof}
{\color{black}
Suppose that an execution has precisely $r$ consistency violations, and that the adversary is $g_2(r)$-bounded. Let $t$ be as in $(\dagger_{d,\ell})$. If $r=0$, then $g_2(0)=\rho_L\leq\rho_C$, and the usual liveness guarantee applies without recovery. So, suppose that $r\geq1$, and let $v_0$ be the first view with a correct leader in the $r^{\text{th}}$ execution of the recovery procedure. By Lemma~\ref{vbound}, all correct processes begin $\mathcal{E}_{r+1}$ within $O(v_0\Delta^*)$ timeslots of $t$. Since $v_0\leq f_a+1$ and $f_a\geq1$, we may therefore take $d=O(f_a\Delta^*)$.

At least $x_rn$ faulty processes have been removed to form $\Pi_{r+1}$. Since $g_2(r)\leq x_r+\rho_L(1-x_r)$, fewer than a $\rho_L$ fraction of the remaining processes are faulty. Since $\rho_L\leq\rho_C$, $\mathcal{E}_{r+1}$ is consistent and live, and no further recovery can occur. Suppose a correct process receives a transaction $\mathtt{tr}$ at timeslot $t'$. Either $\mathtt{tr}$ is in the new genesis log, or it is supplied to $\mathcal{E}_{r+1}$ by $\max\{t+d,t'\}$, according to the instructions in Section~\ref{formalspec}. Since $\ell$ is a common liveness parameter for the executions of $\mathcal{P}$ started by the wrapper, it follows that every correct process finalizes $\mathtt{tr}$ by $\max\{t+d,\mathrm{GST},t'\}+\ell$.\footnote{\textcolor{black}{All correct processes start $\mathcal{E}_{r+1}$ before $t+d$, so $\max\{t+d,\mathrm{GST}\}$ serves as its effective stabilization time.}} This establishes the first claim, since $f_a\geq1$ when $r\geq1$.

For the probabilistic bound, recall that the first correct process in $\Pi^*$ is also the first correct process in every induced permutation $\Pi_r^*$, since the recovery procedure removes only faulty processes. Moreover, its position cannot increase. So, if the first $v$ leaders are faulty in any execution of the recovery procedure, then the first $v$ processes in $\Pi^*$ must be faulty. Since the adversary is static and $(1-\rho_L)$-bounded, the probability of this event is at most $(1-\rho_L)^v$. For $0<\varepsilon<1$, we can therefore choose $v=O(1+\log(1/\varepsilon))$ so that this probability is at most $\varepsilon$. By Lemma~\ref{vbound}, this establishes the claimed probabilistic recovery time $O(\Delta^*\log(1/\varepsilon))$, as $\varepsilon\to0$.\par}
\end{proof}

\section{Related Work} 

\noindent \textbf{Positive results}.  A sequence of papers, including Buterin and
Griffith~\cite{buterin2017casper}, Civit et
al.~\cite{civit2021polygraph}, and Shamis et al.~\cite{shamis2022ia},
describe protocols that satisfy accountability.  Sheng
et al.~\cite{sheng2021bft} analyze accountability for well-known
permissioned protocols such as HotStuff~\cite{yin2019hotstuff}, PBFT
\cite{castro1999practical}, Tendermint
\cite{buchman2016tendermint,buchman2018latest}, and Algorand
\cite{chen2018algorand}. Civit et
al.~\cite{civit2022crime,civit2023easy} describe generic
transformations that take any permissioned protocol designed for the
partially synchronous setting
and provide a corresponding accountable version.
These papers do not describe how to reach consensus on which guilty parties to remove in the event of a consistency violation (i.e.\ how to achieve \textcolor{black}{``recovery''}), and thus fall short of our goals here.  One exception to this point is the ZLB
protocol of Ranchal-Pedrosa and Gramoli~\cite{ranchal2020zlb}, but the ZLB protocol only achieves recovery if the adversary controls
less than a $5/9$ fraction of participants, and does not achieve bounded rollback. Freitas de~Souza et al.~\cite{de2021accountability} also describe a process for removing guilty parties in a protocol for lattice agreement (this abstraction is
weaker than SMR/consensus and can be implemented in an asynchronous system), but their protocol assumes an honest majority and the paper does not consider bounded rollback. 
Sridhar et
al.~\cite{sridhar2023better} specify a ``gadget'' that can be
applied to blockchain protocols operating in the synchronous setting
to reboot and maintain consistency after an attack, but they do not describe how to implement recovery and assume that an honest majority is somehow
reestablished out-of-protocol.

Budish et al.~\cite{budish2024economic} consider ``slashing'' in
proof-of-stake protocols in the ``quasi-permissionless'' setting.
Their main positive result is a protocol that, in the same timing
model considered in this paper (with additional guarantees provided
pre-GST message delays are bounded by a known parameter~$\Delta^*$),
guarantees what they call the ``EAAC property''---honest players never
have their stake slashed, and some Byzantine stake is guaranteed to be
slashed following a consistency violation. Budish et
al.~\cite{budish2024economic} do not contemplate repeated 
consistency violations, a prerequisite to the notions of recoverable
consistency and liveness that are central to this paper.
To the extent that it makes sense to compare their ``recovery
procedure'' with our ``wrapper,'' our protocol is superior in several
respects, with worst-case recovery time $O(n\Delta^*)$ (as opposed
to~$O(n^2\Delta^*)$); probabilistic recovery time 
$O(\Delta^*\text{log}\frac{1}{\varepsilon})$, where $\varepsilon$ is
an
error-probability bound (as opposed to
$O(n\Delta^*\text{log}\frac{1}{\varepsilon})$); and rollback
$2\Delta^*$ (as opposed to unbounded rollback).

Gong et al. \cite{tian2025}  consider the task of recovery and achieve results that are incomparable to those in this paper since they consider different timing and fault models: they restrict to the case of ``ABC faults''\footnote{Processes subject to ``ABC faults’’  can only display faulty behaviour so as to cause consistency failures, and not to threaten liveness.} but consider full partial synchrony, meaning that it is not possible to achieve bounded rollback.  Given this setup, their protocol satisfies the condition that after a single execution of the recovery procedure, a 5/9-bounded adversary is unable to cause further consistency and liveness violations, and after two executions of the recovery procedure a 2/3-bounded adversary is  unable to cause further consistency and liveness violations.  The authors also describe sufficient conditions for a
general BFT SMR protocol to allow for ``complete and sound
fault detection’’ in the case of consistency violations,  even when the actual (unknown) number of faulty processes is as large as  $n-2$.

Distinct from our aims here, i.e., the \emph{removal} of guilty parties so as to achieve recovery in the event of a consistency violation, a number of papers consider the ability of protocols to recover from temporary dishonest majorities (without consideration of any mechanism to ensure that the dishonest majority is temporary). Avarikioti et al.~\cite{avarikioti2019bitcoin} establish a formal sense in which Bitcoin is secure under temporary dishonest majority. We note that, since  Bitcoin is dynamically available \cite{lewis2023permissionless}, it cannot be accountable \cite{neu2022availability}, meaning that there can be no mechanism to remove guilty parties (and only guilty parties) in the event of a consistency violation. Badertscher et al.\ \cite{badertscher2024consensus} extend this analysis to consider dynamically available proof-of-work and proof-of-stake protocols more generally and also establish negative results for BFT-style protocols that do not make use of accountability to remove guilty parties (as we do here).

Prior to the study of accountability, Li and Mazieres \cite{li2007beyond} considered how to design BFT protocols that still offer certain guarantees when more than $f$ failures occur. \textcolor{black}{They} describe a protocol called BFT2F which has the same liveness and consistency guarantees as PBFT when no more than $f<n/3$ players fail; with more than $f$ but no more than $2f$ failures, BFT2F prohibits malicious players from making up operations that clients have never issued and prevents certain kinds of consistency violations.

\vspace{0.2cm} 
\noindent \textbf{Negative results}. There are a number of papers that describe negative results relating to accountability and the ability to punish guilty parties in the \textcolor{black}{``permissionless setting''} (for a definition of the permissionless setting see \cite{lewis2023permissionless}). 
Neu et al.~\cite{neu2022availability} 
prove that no protocol operating in the \textcolor{black}{``dynamically available''} setting (where the number of \textcolor{black}{``active''} parties is unknown)
can provide accountability.
The authors then provide an approach to
addressing this limitation by describing a ``gadget'' that checkpoints
a longest-chain protocol. The ``full ledger'' is then live in the
dynamically available setting, while the checkpointed prefix ledger
provides accountability.   Tas et
al.~\cite{tas2022babylon,tas2023bitcoin}  and Budish et al.~\cite{budish2024economic} also prove negative
results regarding the possibility of punishing guilty participants of proof-of-stake protocols before they are able to cash out of their position.

\section{Final Comments} 

For the sake of simplicity, we have presented our wrapper in the \textcolor{black}{``permissioned''} setting (with a fixed and known set of always active participants). However, since the procedure produces certificates that suffice to verify each new genesis log and the set of processes removed after each consistency violation, it can also be applied directly to proof-of-stake protocols in the quasi-permissionless setting \cite{lewis2023permissionless}. Specifically, our procedure can be used to give a practical replacement for  the (proof-of-concept) recovery procedure of Budish, Lewis-Pye and Roughgarden \cite{budish2024economic} within the formal setup they consider. 

While Theorems \ref{ir1}-\textcolor{black}{\ref{ir4}} show senses in which Theorem \ref{mt} is tight, a number of natural questions remain. For example, our recovery procedure implements a synchronous protocol and has recovery time $O(f_a\Delta^*)$. While Theorem \ref{ir1} establishes that some bound on message delays is required if we are to achieve bounded rollback, one might still make use of a \emph{recovery procedure} that does not require synchrony: could such a procedure achieve recovery time $O(f_a\Delta)$ after GST? Also, while our recovery procedure has rollback bounded by $2\Delta^*$, Theorem \ref{ir1} only establishes a lower bound of $\Delta^*$. Is this lower bound tight, or is $2\Delta^*$ optimal? 

\section*{Acknowledgements}

We thank Ittai Abraham for inspiring conversations during the early
phases of this work. The research of the second author at Columbia
University is supported in part by NSF awards
CCF-2006737 and CNS-2212745. The second author is also Head of
Research at a16z Crypto, a
venture capital firm with investments in blockchain protocols.

\bibliographystyle{ACM-Reference-Format}

\end{document}